\documentclass[preprint,aps,amssymb,amsmath,nobibnotes,floatfix,12pt]{revtex4-1}
\usepackage{docs}%
\usepackage{bm}%
\usepackage{etex}
\usepackage[all]{xy}
\usepackage[T1]{fontenc}
\usepackage[utf8]{inputenc}
\usepackage[english]{babel}
\usepackage[english]{varioref}
\usepackage{amsthm}
\usepackage{amsmath}
\usepackage{mhchem}
\usepackage{subfig}
\usepackage{stmaryrd}
\usepackage{amssymb}
\usepackage{amscd}
\usepackage{tensor}
\usepackage{hyperref}
\usepackage{mathtools}
\DeclareMathAlphabet{\mathpzc}{OT1}{pzc}{m}{it}
\usepackage{mathrsfs}
\usepackage{tensor}
\usepackage{wrapfig}
\newtheorem{theorem}{Theorem}

\newtheorem{proposition}{Proposition}
\DeclareMathOperator{\Sgn}{Sgn}
\newcommand\myatop[2]{\genfrac{}{}{0pt}{}{#1\hfill}{#2\hfill}}

\newtheorem{corollary}{Corollary}[theorem]
\usepackage[sort&compress]{natbib}

\begin{document}

\title{On the composition of an arbitrary collection of $SU(2)$ spins:\\ An Enumerative Combinatoric Approach}

\author{J.A. Gyamfi\footnote{Email corresponding author: \href{mailto:jerryman.gyamfi@sns.it}{jerryman.gyamfi@sns.it} }}
\affiliation{Scuola Normale Superiore di Pisa, Piazza dei Cavalieri 7, 56126 Pisa, Italy.}

\author{V. Barone}
\affiliation{Scuola Normale Superiore di Pisa, Piazza dei Cavalieri 7, 56126 Pisa, Italy.}

\date{\today}

\begin{abstract}
The whole enterprise of spin compositions can be recast as simple enumerative combinatoric problems. We show here that enumerative combinatorics (EC)\citep{book:Stanley-2011} is a natural setting for spin composition, and easily leads to very general analytic formulae -- many of which hitherto not present in the literature. Based on it, we propose three general methods for computing spin multiplicities; namely, 1) the multi-restricted composition, 2) the generalized binomial and 3) the generating function methods. Symmetric and anti-symmetric compositions of $SU(2)$ spins are also discussed, using generating functions. Of particular importance is the observation that while the common Clebsch-Gordan decomposition (CGD) -- which considers the spins as distinguishable --  is related to integer compositions, the symmetric and anti-symmetric compositions (where one considers the spins as indistinguishable) are obtained considering integer partitions. The integers in question here are none other but the occupation numbers of the Holstein-Primakoff bosons.
\par The pervasiveness of $q-$analogues in our approach is a testament to the fundamental role they play in spin compositions. In the appendix, some new results in the power series representation of Gaussian polynomials (or $q-$binomial coefficients) -- relevant to symmetric and antisymmetric compositions -- are presented. 

\end{abstract}


\maketitle 

\section{Introduction}
In many areas of physics and chemistry, one is often faced with the daunting task of determining the Clebsch-Gordan decomposition (CGD) series for a given collection of angular momenta. Seemingly unrelated problems like determining how many linearly independent isotropic tensor isomers (particularly useful in determining rotational averages \cite{art:Andrews-1977} of observables, for example) of rank $n$ there are in a $D-$dimensional space, certain variants of the random walk problem \cite{art:Polychronakos-2016}, symmetric exclusion processes\cite{art:Mendonca-2013}, or even black-hole entropy calculations\citep{art:Livine-2005} can all be related to CGD --  not to mention its use in the detection and characterization of entangled states\citep{art:Cohen-2016}. The relevance of CGD is deeply rooted in group theory and manifest in many of its applications.

\par CGD in the case of two arbitrary $SU(2)$ spin representations is very elementary and reported in textbooks on the subject \cite{book:Zare-1988,book:Atkins-2011}. In the general case of $N$ angular momenta, it is possible to repeat the two angular momenta addition scheme over and over, as explained for example in \cite{book:Atkins-2011}. Needless to say, this approach soon becomes unmanageable as the total number of uncoupled angular momenta $N$ increases, and it gets even worse the higher the magnitude of the momenta involved. An analytical way of determining the Clebsch-Gordan decomposition series is therefore ineluctable. Zachos\cite{art:Zachos-1992} has provided an analytical expression in the case of $N$ spin-$(1/2)$s. Curtright, Van Kortryk and Zachos\cite{art:Curtright-2017} have recently considered the case of a $N$ identical spin$-j$ system. Polychronakos and Sfetsos\cite{art:Polychronakos-2016} have also considered the same problem and have shown how to generalize the results -- by means of the statistical mechanical partition function -- to the case of an arbitrary collection of spins. Though in \cite{art:Polychronakos-2016} an explicit expression for the composition of an arbitrary number of two kinds of spin is given, what is missing in the literature is a general analytic expression for the computation of the multiplicities of the distinct resulting momenta appearing in the CGD of an arbitrary collection of spins. And with this article we propose material to fill that gap. Inspired by the connection between the Hilbert space decomposition of simple symmetric exclusion processes (SSEP) and the CGD of $1/2-$spins\cite{art:Mendonca-2013}, we deduce an even more general relation between symmetric exclusion processes and CGD, based on which we derive analytic formulae for spin multiplicities in symmetric and antisymmetric spin compositions. The latter has been discussed in \cite{art:Polychronakos-2016} based on statistical mechanics, namely the grand partition function, without nonetheless the derivation of explicit analytic formulae for the spin multiplicities. Although the statistical mechanical approach is physically intuitive, we remark that one must be cautious in drawing up conclusions or conjectures on the extensive use of the method as it is only viable when the spins are identical (\S \ref{subsec:exclusion_processes}). But such a limitation is easily surmounted with the approach described below. 
\par The basis of our approach is enumerative combinatorics (EC)\citep{book:Stanley-2011}, and it begins with a bijective mapping of the eigenvalues of the $j_z$ component (henceforth, referred to as "$z-$eigenvalue") of each spin to a subset of $\mathbb{N}_0$ (the set of all natural numbers, zero included). The elements of the new basis may be interpreted as the eigenvectors of the number operator in the Holstein-Primakoff transformation\cite{art:Holstein-1940}. The advantage of the EC approach we discuss below is that it enables one to formulate the CGD problem from the most general perspective right from the start: thus, making it possible to formulate very general solutions from which one can derive solutions to limit cases like the CGD of $N$ spin-$j$ (\S \ref{sec:univariate_spin_sys}) -- which is essentially the foremost general case for which one can find in the literature explicit analytic formulae for the spin multiplicities\citep{art:Polychronakos-2016,art:Mendonca-2013,art:Curtright-2017}. Inspired by enumerative combinatoric analysis, we propose here three methods for the CGD of an arbitrary collection of spins; namely, 1) the multi-restricted composition, 2) the generating function, and 3) the generalized binomial methods. All three are related to each other: from the first, we derive the second method -- from which we derive in turn the third. The generating function method has already been discussed in the literature based on quantum statistical mechanics\citep{art:Polychronakos-2016}, however it ought to be expected since the spin composition problem is purely an enumerative combinatoric one, and generating functions have been the cornerstone of EC and number theory since first introduced by Euler\cite{book:Dickson-2005}. The generating functions given here and their corresponding counterparts in \citep{art:Polychronakos-2016} (see also \citep{art:Mendonca-2013}) are related to each other by a proportionality constant which depends on the total spin of the system.
\par Drawing upon the three methods, we derive general expressions for the dimension of the total $z-$eigenvalue invariant subspaces of the system in Sec. \ref{sec:momentum_inv_sub_dim}. This will be followed by derivation of general expressions for the spin multiplicities in Sec. \ref{sec:multiplicities} -- again, based on the three methods. Then, in Sec. \ref{sec:univariate_spin_sys} we discuss the limit case of "univariate" spin systems. Finally, in Sec. \ref{sec:applications} we illustrate some applications of the general formulae derived in the previous sections as to some problems in quantum physics, number theory and statistics. Especially, as a way of illustrating the usefulness of the results in the previous sections, we determine the CGD series for a specific collection of spins in \S \ref{subsec:application_illustration}. Connections between CGD and lattice paths and some counting numbers like the Catalan and Riordan's will be briefly discussed in \S \ref{subsec:counting_numbers}. Then in \S \ref{subsec:exclusion_processes}, we turn our attention to the application of CGD in analyzing simple symmetric exclusion processes, which will lead us to consider the spin-$\infty$ limit and the CGD series for a collection of such spins. Insights gained in \S \ref{subsec:exclusion_processes} then make it possible to draft down a dictionary which allows one to translate simple exclusion process problems into the language of $SU(2)$ representations, and vice versa. This will be discussed in \S \ref{subsec:symmteric_antisymmetric}, where we are also led to formulate the symmetric and antisymmetric spin composition problem from the enumerative combinatoric perspective. We end Sec. \ref{sec:applications} with some applications of our results in number theory, \emph{viz.} multi-restricted compositions, and statistics in \S \ref{subsec:number_thoery} and \S \ref{subsec:dice}. But before all these, we shall look into one particular important problem which has not yet been addressed in the literature: this has to do with the minimum of the $SU(2)$ spin representations resulting from the composition of an initial collection of spins. While the maximum is easy to prove, the minimum is not. We prove this minimum in the general case in Sec. \ref{sec:min_max_J}. 
\par It is remarkable that the spin multiplicities in the CGD are differences between integer compositions while those in the symmetric and antisymmetric compositions are related to differences between integer partitions. As we know, integer compositions are ordered partitions and, thus, the fact that one has to deal with them in CGDs comes as no surprise since the latter is a spin composition scheme whereby identical spins are treated as distinguishable from each other. On the contrary, when it comes to symmetric and antisymmetric compositions -- which result from enforcing the quantum mechanical requirement of indistinguishability of identical spins -- integer partitions become the norm. This suggests a close-knit between spin compositions and the theory of partitions\cite{book:Andrews-1976}, which has made enormous advancements in the last quarter of a millennium. An important tool which has played a significant role in the progress of EC, number theory (and of course, the theory of partitions) are the so-called $q-$analogues. And we make extensive use of them in this paper. Of particular interests are the $q-$binomial coefficients (or Gaussian polynomials) which, as we will see, are essential in the symmetric and antisymmetric spin compositions (\S \ref{subsec:symmteric_antisymmetric}). As an effort to give analytic expressions for the spin multiplicities in the symmetric and antisymmetric compositions, we treat the power series representation of the Gaussian polynomials in the Appendix, where we give some new results. Indeed, we show that the coefficients in the power series (which are known to be restricted partitions) can be expressed as sums of convoluted sums involving the product of two (modified) Heaviside step-functions which we define in the paper (Eq. \eqref{eq:Heaviside}). We also give some recurrence identities in relation to restricted partitions in the appendix. These results may constitute an efficient algorithm to compute restricted partitions.
\par The formulae we give in the paper can be easily implemented computationally, thus they could greatly enhance recent theoretical and computational efforts aimed at limiting the computational cost of running simulations on many-body systems involving a considerable number of angular momenta. In a forthcoming article, for example, we shall have the occasion to expound more on the use of these analytical formulae to better characterize the Hilbert space of static multi-spin Hamiltonians and show how they can be employed to run simulations of EPR and NMR spectra which scale polynomially with the number of spins rather than exponentially. 
\par 

\section{Coupling of an arbitrary collection of angular momenta}\label{sec:min_max_J}
Say we have a finite collection of $SU(2)$ spin representations: $\mathpzc{A} := \{j_1,j_2, \ldots , j_N\}$, and $j_i \in \{\frac{1}{2},1,\frac{3}{2},2, \ldots \}$. The coupling of these angular momenta will give rise to a multiset\cite{art:Singh-2007,book:Stanley-2011} $\mathpzc{E}_{\mathpzc{A}}$ of angular momenta which we indicate as 
				\begin{equation}
				\label{eq:multiset_E}
				\mathpzc{E}_{\mathpzc{A}} := \{\underbrace{J_0, \ldots ,J_0}_{\lambda_0}, \underbrace{J_1 , \ldots , J_1}_{\lambda_1}, \ldots , \underbrace{J_{m}, \ldots , J_{m} }_{\lambda_{m}}\} := \left\lbrace J^{\lambda_0}_0, J^{\lambda_1}_1, \ldots , J^{\lambda_m}_m\right\rbrace
				\end{equation}
where $\lambda_\kappa$ is the multiplicity of $J_\kappa$. 
\par In group theoretical terms, we may argue as follows: The angular momentum $j_i$ is associated with the Hilbert space $\mathcal{H}^{(j_i)}$ of dimension $2j_i + 1$. The Hilbert space $\mathcal{H}$ of the collection is given by the tensor product
				\begin{equation}
				\mathcal{H} = \mathcal{H}^{(j_1)} \otimes \mathcal{H}^{(j_2)} \otimes \cdots \otimes \mathcal{H}^{(j_N)} \ .
				\end{equation}	
This representation of the Hilbert space $\mathcal{H}$ is reducible into distinct irreducible components $\mathcal{H}^{(J_\kappa)}$, of dimension $2 J_\kappa +1$. $\mathcal{H}$ in terms of its irreducible components is therefore
				\begin{equation}
				\mathcal{H} = \underbrace{\mathcal{H}^{(J_0)} \oplus \ldots \oplus \mathcal{H}^{(J_0)}}_{\lambda_0}\oplus \underbrace{\mathcal{H}^{(J_1)} \oplus \ldots \oplus \mathcal{H}^{(J_1)}}_{\lambda_1} \oplus \cdots \oplus \underbrace{\mathcal{H}^{(J_{m})} \oplus \ldots \oplus \mathcal{H}^{(J_{m})}}_{\lambda_{m}} \ .
				\end{equation}							 
We now determine the distinct elements of $\mathpzc{E}_{\mathpzc{A}}$. To begin with, we shall assume henceforth the following ordering of the distinct elements of $\mathpzc{E}_{\mathpzc{A}}$: $J_0 > J_1 > \ldots > J_{m}$. 

\begin{theorem}\label{thm:J_0_and_J_m}
Given the finite multiset $\mathpzc{A}=\{j_1,j_2, \ldots , j_N\}$ of $SU(2)$ spin representations, then the maximum $(J_0)$ and minimum $(J_m)$ of the related coupled representation $\mathpzc{E}_{\mathpzc{A}}$ are given by the expressions:
				\begin{equation}
				\label{eq:J_0}
				J_0 = \max \mathpzc{E}_{\mathpzc{A}} = \sum_i j_i  
				\end{equation}
				\begin{equation}
				\label{eq:J_m}
				J_{m} = \min \mathpzc{E}_{\mathpzc{A}} = \upsilon_{\mathpzc{A}} \cdot H(\upsilon_{\mathpzc{A}}) + \left(1- H(\upsilon_{\mathpzc{A}})\right)\cdot \frac{(2J_0 \hspace{-0.3cm}\mod 2 )}{2}				
				\end{equation}
respectively, where
				\begin{equation}
				\label{eq:upislon_A}
				\upsilon_{\mathpzc{A}} := 2\cdot \max \mathpzc{A} - J_0 \ ,
				\end{equation}
and where $H\left( x \right)$ is the Heaviside step function, defined here to be
				\begin{equation}
				\label{eq:Heaviside}
				H\left( x \right) := \begin{cases}
				0 , \mbox{ if }x < 0\\
				1 , \mbox{ if }x \geq 0 \ .
				\end{cases}
				\end{equation}	
\end{theorem}

\begin{proof}
Eq. \eqref{eq:J_0} is obvious and we shall not consider it further. The same cannot be said about Eq. \eqref{eq:J_m}.
\par To prove Eq. \eqref{eq:J_m}, it is important to keep in mind that $J_m = \min \mathpzc{E}_{\mathpzc{A}}$. Now, from the multiset $\mathpzc{A}=\{j_1,j_2, \ldots , j_N\}$, we consider the submultiset $\widetilde{\mathpzc{A}}$, defined as
				\begin{equation}
				\widetilde{\mathpzc{A}} := \mathpzc{A}\setminus \{\max \mathpzc{A}\} \ ,
				\end{equation}
\emph{i.e.} $\widetilde{\mathpzc{A}}$ is the submultiset obtained after removing the largest element (or one of the largest elements, if its multiplicity is greater than $1$) of $\mathpzc{A}$ from the latter. The coupled representation of $\widetilde{\mathpzc{A}}$ will yield the multiset $\widetilde{\mathpzc{E}}_{\mathpzc{A}}$ where -- in analogy to Eq. \eqref{eq:multiset_E} --
				\begin{equation}
				\label{eq:multiset_E tilde}
				\widetilde{\mathpzc{E}}_{\mathpzc{A}} := \{\underbrace{\widetilde{J}_0, \ldots ,\widetilde{J}_0}_{\widetilde{\lambda}_0}, \underbrace{\widetilde{J}_1 , \ldots , \widetilde{J}_1}_{\widetilde{\lambda}_1}, \ldots , \underbrace{\widetilde{J}_{m'}, \ldots , \widetilde{J}_{m'} }_{\widetilde{\lambda}_{m'}}\} \ .
				\end{equation}
with $\widetilde{J}_l > \widetilde{J}_{l'}$ if $l<l'$, by convention. We now couple each element of $\widetilde{\mathpzc{E}}_{\mathpzc{A}}$ with $\max \mathpzc{A}$ to generate the multiset $\mathpzc{E}_{\mathpzc{A}}$. Since the coupling of $\widetilde{J}_l \ (\in \widetilde{\mathpzc{E}}_{\mathpzc{A}})$ with $\max \mathpzc{A}$ gives rise to the Clebsch-Gordan decomposition series 
				\begin{equation}
				\max \mathpzc{A}+\widetilde{J}_l , \ \max \mathpzc{A}+\widetilde{J}_l - 1, \ldots , \left|\max \mathpzc{A}-\widetilde{J}_l\right| \ ,
				\end{equation}
-- with all resulting momenta ending up as elements of $\mathpzc{E}_{\mathpzc{A}}$ -- we conclude that 
				\begin{equation}
				\label{eq:min_E def}
				\min \mathpzc{E}_{\mathpzc{A}} = \min \left\lbrace \left. \left|\max \mathpzc{A}-\widetilde{J}_l\right| \  \right| \widetilde{J}_l \in \widetilde{\mathpzc{E}}_{\mathpzc{A}} \right\rbrace \ .
				\end{equation}								
It is important to notice at this point that
				\begin{equation}
				\label{eq:inequalities_obs}
				\max \mathpzc{A} \geq \widetilde{J}_{m'} \qquad \mbox{and} \qquad \widetilde{J}_0 \geq \widetilde{J}_{m'}
				\end{equation}
where the equality in the former \emph{may} hold when $\dim \widetilde{\mathpzc{A}} = 1$, while the equality in the latter holds when $\dim \widetilde{\mathpzc{A}} = 1$. Eq. \eqref{eq:inequalities_obs} leads us to consider two possible scenarios:
\subsection{$\max \mathpzc{A} \geq \widetilde{J}_0$}
In this case, it follows from Eq. \eqref{eq:inequalities_obs} that:
				$
				\max \mathpzc{A} \geq \widetilde{J}_0 \geq \widetilde{J}_{m'} \ .
				$
Therefore, given Eq. \eqref{eq:min_E def} and the fact that $J_l > J_{l'}$ if $l < l'$, we conclude that
				\begin{equation}
				\label{eq:maxA_geq_J_0_tilde}
				\begin{split}
				\min \mathpzc{E}_{\mathpzc{A}} = J_m & = \max \mathpzc{A} - \widetilde{J}_0 \\
				& = 2 \cdot \max \mathpzc{A} - J_0 \ ,
				\end{split} 
				\end{equation}
since $ \max \mathpzc{A} + \widetilde{J}_0 = J_0$.
\subsection{$\max \mathpzc{A} < \widetilde{J}_0$}
When this relation is satisfied, then according to Eq. \eqref{eq:inequalities_obs},
				$
				\widetilde{J}_0 > \max \mathpzc{A} \geq  \widetilde{J}_{m'} 
				$, 
which means that there certainly exists a $\widetilde{J}^*_l \in \widetilde{\mathpzc{E}}_{\mathpzc{A}}$ which satisfies Eq. \eqref{eq:min_E def}. Our aim here though, is not to determine  $\widetilde{J}^*_l$ but the difference $\left|\max \mathpzc{A}-\widetilde{J}^*_l\right|$. In passing, we draw the reader's attention to the fact that the elements of $\widetilde{\mathpzc{E}}_{\mathpzc{A}}$ (and also $\mathpzc{E}_{\mathpzc{A}}$) are either all integers or half-integers, so we can talk of the nature of any of the elements as being an integer or half-integer by just referring to any other element. With this last observation in mind, one can easily show that, 
				\begin{equation}
				\left|\max \mathpzc{A}-\widetilde{J}_l\right| \geq \begin{cases}
				0 , \mbox{ if both $\widetilde{J}_0$ and $\max \mathpzc{A}$ are integers or half-integers } \\
				\frac{1}{2} , \mbox{ if only one of the pair $(\widetilde{J}_0, \max \mathpzc{A})$ is half-integer } \ ,
				\end{cases}
				\end{equation}
since all the elements of $\widetilde{\mathpzc{A}}$ are less or equal to $\max \mathpzc{A}$.
Thus, when $\max \mathpzc{A} < \widetilde{J}_0$, 
				\begin{equation}
				\min \mathpzc{E}_{\mathpzc{A}} = J_m =  \begin{cases}
				0 , \mbox{ if both $\widetilde{J}_0$ and $\max \mathpzc{A}$ are integers or half-integers } \\
				\frac{1}{2} , \mbox{ if only one of the pair $(\widetilde{J}_0, \max \mathpzc{A})$ is half-integer } \ ,
				\end{cases}
				\end{equation}
which can be written more succinctly as,
				\begin{equation}
				\label{eq:maxA_lt_J_0_tilde}
				\begin{split}
				J_m & = \frac{1}{2} \left[2\left(\max \mathpzc{A} + \widetilde{J}_0 \right)   \ \mbox{mod} \ 2 \right]\\
				& =  \frac{1}{2} \left[2J_0   \ \mbox{mod} \ 2 \right] \ .
				\end{split}
				\end{equation}
Putting together Eqs. \eqref{eq:maxA_geq_J_0_tilde} and \eqref{eq:maxA_lt_J_0_tilde} yields Eq. \eqref{eq:J_m}.
\end{proof}

Going back to Eq. \eqref{eq:J_m}, we observe that
				\begin{equation}
				J_m = \begin{cases}
				\upsilon_{\mathpzc{A}} , \mbox{ if }\upsilon_{\mathpzc{A}} \geq 0\ , \\
				0 , \mbox{ if }\upsilon_{\mathpzc{A}} < 0 \ \mbox{and }\upsilon_{\mathpzc{A}}\ \mbox{is an integer} \ , \\
				\frac{1}{2} , \mbox{ if }\upsilon_{\mathpzc{A}} < 0 \ \mbox{and }\upsilon_{\mathpzc{A}}\ \mbox{is half-integer} \ .
				\end{cases}
				\end{equation}	
Since $J_0 - J_\kappa = \kappa, \ \forall \ \kappa \in \{0, \ldots, m\}$, it follows that the number of distinct elements of $\mathpzc{E}_{\mathpzc{A}}$, $N'_{\mathpzc{E}_{\mathpzc{A}}}$, is certainly $N'_{\mathpzc{E}_{\mathpzc{A}}} = J_0 - J_m + 1$ ,
and so $m = N'_{\mathpzc{E}_{\mathpzc{A}}} - 1$. As already noted above, if $J_0$ is half-integer, so are all the elements of $\mathpzc{E}_{\mathpzc{A}}$, and the same applies when $J_0$ is an integer.

\section{The  total $z-$eigenvalue invariant subspaces and their dimensions}\label{sec:momentum_inv_sub_dim}
Each angular momentum $j_i$ of $\mathpzc{A}=\{j_1, j_2, \ldots , j_N\}$ has its own set of possible $z-$eigenvalues, which we indicate as $\mathpzc{M}_i$, \emph{i.e.} 
			\begin{equation}
			\mathpzc{M}_i := \{j_i, j_i - 1 , \ldots , -j_i\} \ , \qquad i \in \{1, 2, \ldots , N\}\ .
			\end{equation}
The cardinality of $\mathpzc{M}_i$, $\dim \mathpzc{M}_i$, is obviously $(2j_i+1)$. To make contact with EC, we note that elements of the set $\mathpzc{M}_i$ can also be represented as
			\begin{equation}
			\label{eq:M_to_n}
			\mathpzc{M}_i := \left\lbrace j_i-n_i \left. \right| \ n_i \in \mathbb{N}_0 \ \wedge \  0 \leq n_i \leq 2j_i \right\rbrace  \ .
			\end{equation}
For a given $j_i$, we denote the set of all possible $n_i$ according to Eq. \eqref{eq:M_to_n} as $\mathpzc{N}_i$; \emph{i.e.} $\mathpzc{N}_i := \{0, 1, 2, \ldots , 2j_i\}$. Evidently, there is a one-to-one correspondence between the set $\mathpzc{N}_i$ and the set $\mathpzc{M}_i$. Due to this bijective relation, we can perform operations with $\mathpzc{N}_i$ instead of $\mathpzc{M}_i$. This is very convenient because -- unlike $\mathpzc{M}_i$ -- the set $\mathpzc{N}_i$ is always a subset of $\mathbb{N}_0$, independent of whether $j_i$ is half-integer or not. Furthermore, we point out that $n_i$ is an eigenvalue of the number operator for the Holstein-Primakoff bosons of the single spin $j_i$.
\par Each distinct element $J_l$ of the coupled representation of $\mathpzc{A}$, \emph{i.e.} $\mathpzc{E}_{\mathpzc{A}}$, also has its own set $\mathpzc{M}'_l$ of possible total $z-$eigenvalues, and -- in analogy to Eq. \eqref{eq:M_to_n} -- we can write
			\begin{equation}
			\mathpzc{M}'_l := \left\lbrace J_l-n_l \left. \right| \ n_l \in \mathpzc{N}_l \right\rbrace \ ,
			\end{equation}			  
where $\mathpzc{N}_l := \{0, 1, 2, \ldots, 2J_l\}$. The multiset sum of the various $\mathpzc{M}'_l$ yields the multiset $\mathpzc{M}'$, that is,
			\begin{equation}
			\label{eq:def_M_prime}
			\mathpzc{M}' := \biguplus^m_{l=0} \ \mathpzc{M}'_l\ .
			\end{equation}
$\mathpzc{M}'$ contains as elements all possible total $z-$eigenvalues in the coupled representation, each repeated a certain number of times $\Omega$. The multiplicity $\Omega$ of each total $z-$eigenvalue reflects the number of possible $J_l \in \mathpzc{E}_{\mathpzc{A}}$ from which it can originate. 
\par To determine these multiplicities, what we do is to first make use of the fact that the $z-$eigenvalues in the coupled representation are simply given by the sum of their uncoupled counterparts. On that account, it follows from Eq. \eqref{eq:M_to_n} that
			\begin{equation}
			\label{eq:multiset_M_coup_rep}
			\begin{split}
			\mathpzc{M}' & =\left\lbrace \left. \sum_i \left( j_i - n_i\right)  \right| j_i \in \mathpzc{A} \ , \ n_i \in \mathpzc{N}_i \right\rbrace \\
			& = \left\lbrace \left.   J_0 -\sum_i n_i  \right|  \ n_i \in \mathpzc{N}_i \right\rbrace
			\end{split}
			\end{equation}
where we have made use of Eq. \eqref{eq:J_0}. Given that the sum of a finite number of natural numbers is also obviously a natural number, we can rewrite Eq. \eqref{eq:multiset_M_coup_rep} as
			\begin{equation}
			\label{eq:M_pr}
			\mathpzc{M}' = \left\lbrace \left.   J_0 - n  \right|  n \in \mathpzc{N} \right\rbrace
			\end{equation}			 
where,
			\begin{subequations}
			\label{eq:deff_n_N}
			\begin{align}
			n & := \sum_i n_i\ , \qquad n_i \in \mathpzc{N}_i \label{eq:deff_n_N_a} \\
			\mathpzc{N} & := \left\lbrace 0 , 1 , \ldots , \sum_i \ \max \mathpzc{N}_i \right\rbrace = \left\lbrace 0 , 1 , \ldots , 2 J_0 \right\rbrace \ .
			\end{align}
			\end{subequations}
It is therefore straightforward to see that the multiplicity of the element $(J_0 - n)$ of the multiset $\mathpzc{M}'$ is just the number of compositions of the integer $n$ according to Eq. \eqref{eq:deff_n_N_a}. In fact, Eq. \eqref{eq:deff_n_N_a} defines a "multi-restricted composition" of $n$, whereby in addition to the number of parts restricted to $N$, each part $n_i$ is also restricted to the set $ \mathpzc{N}_i$. Contrary to common compositions, the number zero is an admissible part here for the sake of mathematical consistency, though the results of the computations which follow remain unchanged if one ignores it. We shall come back to this point later on.
\par We shall assume the notation $M_n := J_0 -n$. And so there is a one-to-one correspondence between $\{n\}$ and the set of distinct total spin orientations $\{M_n\}$. 
\begin{theorem}
Let  $\Omega_{\mathpzc{A},n}$ represent the number of elements of the multiset $\mathpzc{E}_{\mathpzc{A}}$ which admit $M_n$ as a possible eigenvalue of its $z-$component. Then,  $\Omega_{\mathpzc{A},n}$ is given by the following expression:
				\begin{equation}
				\label{eq:dim_B(n)}
				\boxed{
				\Omega_{\mathpzc{A},n} = \sum_{A(n)\  \in \ P(\mathpzc{A};n) } \ \prod^{\dim \widetilde{A}(n)-1}_{\nu=0} 
				\binom{\omega_{n}(\nu) - (1-\delta_{\nu,0}) \sum^{\nu-1}_{l=0} s_{n}(l)}{s_{n}(\nu)} }
				\end{equation}
where,
	\begin{itemize}
	\item $P(\mathpzc{A};n) := $ the set of all (unordered) multi-restricted partition of the integer $n$ into at most $\dim \mathpzc{A}=N$ parts, with the value of each part being at most only one of the $N$ values of $\{2j_1,\ldots , 2j_N \}$, where $j_i \in \mathpzc{A}$ and with $0$ as an admissible part; \\
	\item $A(n) :=$ an element of $P(\mathpzc{A};n)$; the sum  in Eq. \eqref{eq:dim_B(n)} is over all $A(n)$;\\
	\item $\widetilde{A}(n) := $ the set of distinct elements (or parts) of $A(n)$. $\nu$ runs over all elements $\{a_{n}(\nu)\}$ of $\widetilde{A}(n)$: with $a_{n}(0)$ as the largest integer, $a_{n}(1)$ the second largest integer, and so forth;\\
	\item $\omega_{n}(\nu) := $ the number of $2j_i$ greater or equal to $a_{n}(\nu)$;
	\item $s_{n}(\nu) := $ the number of times $a_{n}(\nu)$ appears in $A(n)$.
	\end{itemize} 
\end{theorem}
\begin{proof}
It is clear that for any given finite multiset $\mathpzc{A}$ we can define the multiset $\mathpzc{N}_{max}$, where
				\begin{equation}
				\mathpzc{N}_{max} (\mathpzc{A}) := \left\lbrace \max  \mathpzc{N}_i  \right\rbrace_{\mathpzc{A}} = \{2j_1, 2j_2, \ldots , 2j_N\}\ .
				\end{equation}
Say $P(\mathpzc{A};n)$ the set of all (unordered) multi-restricted partitions of the integer $n$ into $\dim \mathpzc{A}(=N)$ parts, with the restriction that the $i-$th part, $n_i$, belongs to the set $\mathpzc{N}_i$. That is,
				\begin{equation}
				P(\mathpzc{A};n) := \{A(n;1), A(n;2), \ldots , A(n;q_n) \}
				\end{equation}
where $A(n;k)$ is the $k-$th multi-restricted partition of $n$ in reference to $\mathpzc{A}$ according to the description given above; $q_n := \dim P(\mathpzc{A};n) $. For convenience, we choose to write the elements of $P(\mathpzc{A};n)$ as multisets, \emph{i.e.}
				\begin{equation}
				A(n;k) := \{ a^{s_n(k,0)}_{k,0}, a^{s_n(k,1)}_{k,1} , \ldots , a^{s_n(k,\mu_{k,n})}_{k,\mu_{k,n}}\}
				\end{equation}
where the $a_{k,\nu}$s are positive integers (in other words, parts of the restricted partition of $n$); $s(k,\nu)$ is the multiplicity of the part $a_{k,\nu}$ and $\mu_{k,n} := \dim A(n;k) -1$. Evidently,
				\begin{equation}
				n= \sum^{\mu_{k,n}}_{\nu=0} s_n(k,\nu) \cdot a_{k,\nu} \ , \qquad \ N = \sum^{\mu_{k,n}}_{\nu=0} s_n(k,\nu) \ .
				\end{equation}
Finally, let us denote the set of all distinct elements of $A(n;k)$ as $\widetilde{A}(n;k)$:
				\begin{equation}
				\widetilde{A}(n;k) := \{ a_{k,0}, a_{k,1} , \ldots , a_{k,\mu_{k,n}}\} \ .
				\end{equation}
To determine the composition $\Omega_{\mathpzc{A},n}$ of $n$ as described above according to \eqref{eq:deff_n_N_a}, we may consider the elements of $\mathpzc{A}$ as unconnected channels ($N$ in total) among which we need to distribute \emph{successfully} the elements of $A(n;k)$. The caveats here are: 1) the positive integers the channel $i$ can accommodate cannot exceed $\max \mathpzc{N}_i$; 2) given the multiset $A(n;k)$, each channel can accommodate just an element of $A(n;k)$; 3) by "successfully" we mean the elements of $A(n,k)$ must be distributed among the $N$ channels without leaving any behind, according to the preceding caveats. Given $A(n;k)$, we denote the number of ways the distribution can be successfully achieved as $\varrho[A(n;k)]$. It is evident then that
				\begin{equation}
				\label{eq:Omega_n_a}
				\Omega_{\mathpzc{A},n} = \sum_k \varrho[A(n;k)] \ , \qquad \ A(n;k) \in P(\mathpzc{A};n) . 
				\end{equation}
How do we then determine $\varrho[A(n;k)]$? In order to ensure a successful distribution of the elements of $A(n;k)$ among the $N$ channels, we must distribute first $\max \widetilde{A}(n;k)$, followed by the second largest integer of $\widetilde{A}(n;k)$, and so on. To simplify matters, we assume that $a_{k,\nu} > a_{k,\nu'}$ if $\nu < \nu'$. Thus, $\max \widetilde{A}(n;k) = a_{k,0}$. If we denote the number of ways of distributing the number $a_{k,\nu} \in \widetilde{A}(n;k) $ as $D_n(k,\nu)$, then it follows that
				\begin{equation}
				D_n(k,0) = \binom{
				\omega_n(0)}{
				s_n(k,0)
				} \ ,
				\end{equation}
where $\omega_n(\nu)$ is the number of channels which can successfully take in $a_{k,\nu}$. In general, $\omega_n(\nu)$ is the number of elements of the multiset $\mathpzc{N}_{max}(\mathpzc{A})$ which are not less than $a_{k,\nu}$:
				\begin{equation}
				\omega_n(\nu) = \sum^N_{i=1}  H(2j_i-a_{k,\nu}) \, \qquad \ 2j_i \in \mathpzc{N}_{max}(\mathpzc{A}) \ ,
				\end{equation}
where $H(x)$ is the Heaviside step function defined in Eq. \eqref{eq:Heaviside}.			 
\par We now move on to the second largest integer of $\widetilde{A}(n;k)$: $a_{k,1}$. The number of channels which can take in this integer is $\omega_n(1)$, which includes the $\omega_n(0)$ channels since $a_{k,0} > a_{k,1}$. But of the $\omega_n(0)$ channels,  $s_n(k,0)$ of them have already been occupied by $a_{k,0}$. We are therefore left effectively with $\omega_n(1)-s_n(k,0)$ channels available to accommodate the $s_n(k,1)$ integers of value $a_{k,1}$. Hence,
				\begin{equation}
				D_n(k,1) = \binom{
				\omega_n(1) - s_n(k,0)}{
				s_n(k,1)
				} \ .
				\end{equation}
Following similar arguments, for $a_{k,2}$ we shall have at our disposal $\omega_n(2)-s_n(k,0)-s_n(k,1)$ channels to distribute $s_n(k,2)$ of it. And so,
				\begin{equation}
				D_n(k,2) = \binom{
				\omega_n(1) - s_n(k,0) - s_n(k,1) }{
				s_n(k,2)
				} \ .
				\end{equation}
\par We easily infer from the above arguments that, for a given $\mathpzc{A}$,
				\begin{equation}
				\label{eq:Omega_n_b}
				D_n(k,\nu) = \binom{
				\omega_n(\nu) - (1-\delta_{0, \nu}) \sum^{\nu-1}_{l=0}s_n(k,l)}{
				s_n(k,\nu)
				} 	\ .		
				\end{equation}
It is easy to realize that for a given $A(n;k)$, $ \varrho[A(n;k)]$ is the product of the various $D_n(k,\nu)$, \emph{i.e.}
				\begin{equation}
				\label{eq:Omega_n_c}
				 \varrho[A(n;k)] = \prod^{\dim \widetilde{A}(n;k)-1}_{\nu=0} D_n(k,\nu) \ .
				\end{equation}
Putting together Eqs. \eqref{eq:Omega_n_a}, \eqref{eq:Omega_n_c} and \eqref{eq:Omega_n_b}, and simply writing $s_n(k, \nu) \to s_n(\nu)$, $A(n;k) \to A(n)$, $\widetilde{A}(n;k) \to \widetilde{A}(n)$, we get Eq. \eqref{eq:dim_B(n)}.
\end{proof} %
We shall refer to Eq. \eqref{eq:dim_B(n)} as the \emph{multi-restricted composition method} for determining $\{\Omega_{\mathpzc{A},n}\}$. Interestingly, $\Omega_{\mathpzc{A},n}$ may be interpreted as the number of ways of distributing $n$ Holstein-Primakoff bosons\citep{art:Holstein-1940} among $N$ spins, knowing that the $i-$th spin can accommodate at most $2j_i$ of such bosons.
\par One can now understand why the parts with value $n_i=0$ do not contribute to $\Omega_{\mathpzc{A},n}$: the effective number of ways of distributing any allowed multiplicity of the number zero among the channels, according to the above criteria, is always $1$. Furthermore, it is easy to realize  that when $\mathpzc{A}=\left\lbrace \frac{1}{2}^N\right\rbrace$ -- \emph{i.e.} in the case of $N$ spin-$\frac{1}{2}$s, Eq. \eqref{eq:dim_B(n)} reduces to
				\begin{equation}
				\label{eq:spin_1/2_Omega_limit}
				\Omega_{\mathpzc{A},n} = \binom{
				N }{
				n
				} \ .
				\end{equation}
\par An important property of the integers $\{\Omega_n\}$ (in the following, we shall on some occasions simply write $\Omega_n$ instead of $\Omega_{\mathpzc{A},n}$) is that
				\begin{equation}
				\label{eq:symmetric_Omega}
				\Omega_{n} = \Omega_{2J_0-n} \ , \ \qquad \forall n \in \mathpzc{N} \ ,
				\end{equation}
which implies a reciprocal distribution of these integers. The relation in Eq. \eqref{eq:symmetric_Omega} holds because $\Omega_n$ and $\Omega_{2J_0-n}$ are the dimensions of the subspaces of total $z-$eigenvalues $M_n$ and $M_{2J_0-n}$, respectively. But from Eq. \eqref{eq:M_pr}, we have that   $\left| M_n \right| = \left| M_{2J_0-n}\right|$, \emph{i.e.} the two subspaces are related by T-symmetry -- which necessarily implies that they must have the same dimension. Moreover, 
				\begin{equation}
				\label{eq:max_Omega_n}
				\max \ \{\Omega_n \} = \Omega_{\lfloor J_0 \rfloor} \ .
				\end{equation}
where $\lfloor \bullet \rfloor$ is the floor function. We add that $\{\Omega_n \}$ is a multiset and so the subspace characterized by $n= {\lfloor J_0 \rfloor}$ may not be the only subspace of maximum dimension.
\subsection{The generating function $G_{\mathpzc{A},\Omega}(q)$ for $\{\Omega_{\mathpzc{A},n}\}$}\label{subsec:gen_mom_inv_sub_dim}
\begin{theorem}
The generating function $G_{\mathpzc{A},\Omega}(q)$ for the integers $\{\Omega_{\mathpzc{A},n}\}$ is
				\begin{equation}
				\label{eq:gen_Omega_n}
				G_{\mathpzc{A},\Omega}(q)  = \prod_\alpha \left[ \Lambda_\alpha(q) \right]^{N_\alpha}=\sum^{2J_0}_{n=0} \Omega_{\mathpzc{A},n} \ q^n 
				\end{equation}			
where the index $\alpha$ runs over distinct elements $\{j_\alpha\}$ of $\mathpzc{A}$, $N_\alpha$ is the multiplicity of $j_\alpha$ in $\mathpzc{A}$, and $\Lambda_\alpha(q)$ is the polynomial function
				\begin{equation}
				\label{eq:def_Lambda}
				\Lambda_\alpha (q) := 1 + q + q^2 + \ldots + q^{2j_\alpha} \ .
				\end{equation}
\end{theorem}
\begin{proof}
To prove Eq. \eqref{eq:gen_Omega_n}, it is advantageous to recall that each $\Omega_n$ characterizes the dimension of the invariant subspace of constant $M_n \in \mathpzc{M}'$. The sum total of the dimension of these subspaces must necessarily be the same as the tensor space of the spaces of $\mathpzc{A}$ elements. This yields us the identity
				\begin{equation}
				\label{eq:dim_identity}
				\prod_\alpha (2j_\alpha +  1)^{N_\alpha} = \sum^{2J_0}_{n=0} \Omega_{\mathpzc{A},n} \ .
				\end{equation}
The LHS of Eq. \eqref{eq:dim_identity} can be seen as the result of the limit 
				\begin{equation}
				\lim_{q \to 1} \ \sum^{2J_0}_{n=0} \Omega_{\mathpzc{A},n} \ q^n  \ .
				\end{equation}
We may thus seek to find the polynomial $G_{\mathpzc{A},\Omega}(q)$ such that
				\begin{equation}
				\label{eq:def_G_Omega_n}
				\lim_{q \to 1}\  G_{\mathpzc{A},\Omega}(q) = \prod_\alpha (2j_\alpha +  1)^{N_\alpha} \ .
				\end{equation}
It is manifest from Eq. \eqref{eq:def_G_Omega_n} that the latter is satisfied if there exists a polynomial in $q$, $g_{\mathpzc{A},\Omega,\alpha}(q)$, such that 
				\begin{equation}
				\label{eq:g_Omega_alpha_a}
				\lim_{q \to 1} \ g_{\mathpzc{A},\Omega,\alpha}(q) = 2j_\alpha + 1 \ .
				\end{equation}
$ g_{\mathpzc{A},\Omega,\alpha}(q)$ must have positive coefficients and be of degree $2j_\alpha$ so as to ensure that $\deg \left[ \ G_{\mathpzc{A}, \Omega}(q)\right] = \deg \left[\ \sum^{2J_0}_{n=0} \Omega_n \ q^n \right] = 2J_0$. From these conditions, we conclude that all the coefficients of the polynomial $ g_{\mathpzc{A},\Omega,\alpha}(q)$ must be of value $1$. But this implies that $g_{\mathpzc{A},\Omega,\alpha}(q)$ is none other but the $q-$analogue of the integer $2j_\alpha + 1$, \emph{i.e.}
				\begin{equation}
				\label{eq:g_Omega_alpha_b}
				g_{\mathpzc{A},\Omega,\alpha}(q) = \left[ 2j_\alpha + 1\right]_q \ ,
				\end{equation}
where by definition\cite{book:Kac-2001},
				\begin{equation}
				\label{eq:q-analogue_expans}
				\left[ n \right]_q := \frac{q^n - 1}{q-1}  = 1 + q + q^2 + \ldots + q^{n-1} \ .
				\end{equation}
Accordingly, it follows from Eqs. \eqref{eq:def_G_Omega_n}, \eqref{eq:g_Omega_alpha_a} and \eqref{eq:g_Omega_alpha_b} that
				\begin{equation}
				\label{eq:gen_Omega_n_c}
				G_{\mathpzc{A},\Omega}(q) =  \prod_\alpha \left(\left[2j_\alpha +  1\right]_q\right)^{N_\alpha} \ .
				\end{equation}
As should be expected, $\deg \left[ \ G_{\mathpzc{A}, \Omega}(q)\right] = \deg \left[\ \sum^{2J_0}_{n=0} \Omega_n \ q^n \right] = 2J_0$.
\end{proof}
Going back to the identity in Eq. \eqref{eq:symmetric_Omega}, we conclude that $G_{\mathpzc{A},\Omega}(q)$ is a \emph{reciprocal} polynomial. But that is not all: it is also \emph{unimodal}, owing to Eq. \eqref{eq:max_Omega_n} and the fact that $\Omega_n \leq \Omega_{n'}$, for $ n< n'$ and $ 0\leq \{n,n'\} \leq \lfloor{J_0\rfloor}$.
\par The generating function $G_{\mathpzc{A},\Omega}(q)$ allows an alternative and more analytic expression for the integers $\Omega_{\mathpzc{A},n}$. As it happens, the following theorem holds:
\begin{theorem}\label{thm:generalized_dim_B(n)}
Given the multiset $ \mathpzc{A} $, the integer $ \Omega_{\mathpzc{A},n} $ -- as defined above -- satisfies the relation,
				\begin{equation}
				\label{eq:generalized_dim_B(n)}
				\boxed{
				\Omega_{\mathpzc{A},n} = \sum_{\substack{\sum^{\sigma}_{\alpha = 1} (2j_\alpha + 1)s_\alpha \leq n \\ 0 \leq s_\alpha \leq N_\alpha}} (-1)^{s_1 + s_2 + \ldots + s_\sigma} 
				\binom{
				N + n - 1 - \sum^{\sigma}_{\alpha = 1} (2j_\alpha + 1)s_\alpha}{
				N- 1
				} 
				\binom{
				N_1}{
				s_1
				}\ldots
				\binom{
				N_\sigma }{
				s_\sigma
				}  }
				\end{equation}
where $ \sigma $ is the number of distinct angular momenta present in the multiset $\mathpzc{A}$; the index $ \alpha $ runs over all distinct elements of $ \mathpzc{A} $; $j_\alpha$ is the $\alpha-$th distinct element of $\mathpzc{A}$ with multiplicity $N_{\alpha}$ and $ 0 \leq n \leq \sum_{\alpha} 2N_{\alpha} j_{\alpha} = 2J_0$.
\end{theorem}

\begin{proof}
The prove is very simple. Surely, from Eqs. \eqref{eq:gen_Omega_n_c} and \eqref{eq:q-analogue_expans} we may rewrite $G_{\mathpzc{A},\Omega}(q)$ as
					\begin{equation}
					G_{\mathpzc{A},\Omega}(q) = (1-q)^{-N} \prod_\alpha \left(1-q^{2j_\alpha + 1} \right)^{N_\alpha} \ .
					\end{equation}
After expanding the factor  with negative exponent on the RHS as a negative binomial series (see for example \citep{book:Riordan-2002}) and the other factors by the normal binomial theorem, and collecting terms with the same $q^n$, we get the result in Eq. \eqref{eq:generalized_dim_B(n)}.
\end{proof}
Eq. \eqref{eq:generalized_dim_B(n)} is a generalization of a result in \cite{art:Bollinger-1993}, and it is readily implementable computationally as to Eq. \eqref{eq:dim_B(n)}. As far as we know, Eq. \eqref{eq:generalized_dim_B(n)} is also the most concise general analytical formula for calculating $\Omega_{\mathpzc{A},n}$. The summation reduces to very few terms, given the condition that $\sum^{\sigma}_{\alpha = 1} (2j_\alpha + 1)s_\alpha \leq n$. We shall term Eq. \eqref{eq:generalized_dim_B(n)} as the \emph{generalized binomial method} for determining $\Omega_{\mathpzc{A},n}$.

\section{The multiplicities $\{\lambda_{\mathpzc{A},\kappa}\}$ and corresponding generating function $G_{\mathpzc{A},\lambda}(q)$}\label{sec:multiplicities}
Given that the multiset $\mathpzc{M}'$ (see Eq. \eqref{eq:def_M_prime}) is generated by $\mathpzc{E}_{\mathpzc{A}}$ (see Eq. \eqref{eq:multiset_E}) and $\Omega_{\mathpzc{A},n}$ is the dimension of the invariant subspace of constant $M_n \in \mathpzc{M}'$, then 
				\begin{equation}
				\label{eq:Omega_lambda_rel}
				\Omega_{\mathpzc{A},n} = \sum^n_{\kappa=0} \lambda_{\mathpzc{A},\kappa} \ , \qquad \ n \in \{0,1, \ldots , m\}
				\end{equation}				 
since if $J_l < \left| M_{n} \right|$, then $M_{n} \not\in \mathpzc{M}'_l$ (recall also that according to the convention employed in this article, $J_l > J_{l'}$ if $l < l'$). From Eq. \eqref{eq:Omega_lambda_rel}, we derive that the multiplicity $\lambda_{\mathpzc{A},\kappa}$ of the distinct element $J_\kappa$ of $\mathpzc{E}_{\mathpzc{A}}$ is therefore given by the relation
				\begin{equation}
				\label{eq:lambda_recursive}
				\lambda_{\mathpzc{A},\kappa} = \Omega_{\mathpzc{A},\kappa} - (1-\delta_{\kappa , 0}) \  \Omega_{\mathpzc{A},\kappa - 1} \ , \qquad \kappa \in \{0,1, \ldots , m\} \ .
				\end{equation}
Although this relation holds independent of the method employed in determining $\{\Omega_{\mathpzc{A},n} \}$, we shall refer to it as the multi-restricted composition way of determining $\lambda_{\mathpzc{A},\kappa}$.
An elementary analysis of the recursive relation in Eq. \eqref{eq:lambda_recursive} shows that the generating function $G_{\mathpzc{A},\lambda}(q)$ for $\{\lambda_{\mathpzc{A},\kappa}\}$ must be related to $G_{\mathpzc{A},\Omega}(q)$ through the equation
				\begin{equation}
				\label{eq:gen_lambda_n}
				\boxed{
				G_{\mathpzc{A},\lambda}(q) = (1-q) \ G_{\mathpzc{A},\Omega}(q) = \sum^{2J_0+1}_{\kappa=0} \lambda_{\mathpzc{A},\kappa} \ q^\kappa } \ ,
				\end{equation}				 
which is the generating function way of determining $\lambda_{\mathpzc{A},\kappa}$. The positive coefficients of the generating function $ G_{\mathpzc{A},\lambda}(q)$ give the multiplicities $\{\lambda_{\mathpzc{A},\kappa}\}$ in Eq. \ref{eq:multiset_E}. Nevertheless, a few observations are due here: Unlike $G_{\mathpzc{A},\Omega}(q)$, $\deg \left[ G_{\mathpzc{A},\lambda}(q)\right]= 2J_0 + 1$. Moreover, exploiting the fact that $G_{\mathpzc{A},\Omega}(q)$ is reciprocal and unimodal, it can be easily proved that 
				\begin{equation}
				\lambda_{\kappa} = - \lambda_{2J_0+1-\kappa} \ ,
				\end{equation}
which -- together with the fact that $G_{\mathpzc{A},\lambda}(q=1)=0, \forall \mathpzc{A}$ -- implies that out of the $2(J_0+1)$ terms of $G_{\mathpzc{A},\lambda}(q)$, only an even number of them has nonzero coefficients. 	This number is precisely $2(m+1)$. We call those terms with zero coefficients the "sinking terms". We see that when $m=J_0$, there are no sinking terms in $G_{\mathpzc{A},\lambda}(q)$. In general, the number of sinking terms one should expect is $2J_m$.
\par Similar to Eq. \eqref{eq:generalized_dim_B(n)} of theorem \ref{thm:generalized_dim_B(n)}, the following theorem can be stated:
\begin{theorem}
Let $\mathpzc{A}$ be a multiset of angular momenta and $\mathpzc{E}_{\mathpzc{A}}$ its corresponding coupled representation multiset (see Eq. \eqref{eq:multiset_E}). Then, the multiplicity $\lambda_{\mathpzc{A},\kappa}$ of $J_\kappa \in \mathpzc{E}_{\mathpzc{A}}$ is given by the following expression:
				\begin{equation}
				\label{eq:generalized_lambda}
				\boxed{
				\lambda_{\mathpzc{A},\kappa} = \sum_{\substack{\sum^{\sigma}_{\alpha = 1} (2j_\alpha + 1)s_\alpha \leq \kappa \\ 0 \leq s_\alpha \leq N_\alpha}} (-1)^{s_1 + s_2 + \ldots + s_\sigma} 
				\binom{N + \kappa - 2 - \sum^{\sigma}_{\alpha = 1} (2j_\alpha + 1)s_\alpha}{N- 2}
				\binom{N_1 }{s_1}\ldots
				\binom{N_\sigma}{s_\sigma} }
				\end{equation}
where $0 \leq \kappa \leq m$.
\end{theorem}
The proof follows the same line of reasoning as that of theorem \ref{thm:generalized_dim_B(n)}, so we leave it as an exercise for the reader. Eq. \eqref{eq:generalized_lambda} is the generalized binomial method for determining $\lambda_{\mathpzc{A},\kappa}$. 
 
\section{Limit case: univariate spin systems.}\label{sec:univariate_spin_sys}				
A collection of $N$ spins is said to be an \emph{univariate spin system (USS)}  if $j_1 = j_2=\ldots = j_N=j$. Note that an USS may not necessarily be a system of \emph{identical} spins (IS). 
\par The limit case of univariate spin system, \emph{i.e.} when $\mathpzc{A}=\{j^N\}$, has been extensively discussed in the literature \citep{art:Mendonca-2013, art:Polychronakos-2016, art:Curtright-2017} lately, though the distinction between them and systems of identical spins has not been hitherto emphasized. We show here that the results in the above cited literature can be easily deduced from the generalized equations \eqref{eq:dim_B(n)} and \eqref{eq:generalized_lambda}. For convenience, we shall sometimes represent the integers $\{\Omega_{\mathpzc{A},n}\}$ and $\{\lambda_{\mathpzc{A},n}\}$ of $\mathpzc{A}=\{j^N\}$ (\emph{i.e.} an USS) simply as $\Omega_{j^N,n}$ and $\lambda_{j^N,n}$, respectively.
\par Beginning with the integers $\{\Omega_{\mathpzc{A},n}\}$, we see that in the case of an univariate spin system of $N$ spin-$j$, Eq. \eqref{eq:generalized_dim_B(n)} reduces to
				\begin{equation}
				\label{eq:mono_spin_Omega_n}
				\boxed{
				\Omega_{j^N,n}= \sum^{\lfloor{\frac{n}{2j+1}\rfloor}}_{s=0} (-1)^s
				\binom{
				N + n - 1 -  (2j + 1)s }{
				N- 1
				} 
				\binom{
				N }{
				s
				} }\ .
				\end{equation}								 
That is, in the case of an univariate spin system, the integers $\{\Omega_{\mathpzc{A},n}\}$ are none but the entries of an extended Pascal triangle (see \citep{art:Bollinger-1993} and compare Eq. \eqref{eq:mono_spin_Omega_n} with \citep[Eq. A.6]{art:Mendonca-2013}). 
\par \begin{proposition} In terms of generalized hypergeometric functions\cite{book:Koepf-2014}, it can be easily shown that Eq. \eqref{eq:mono_spin_Omega_n} yields
				\begin{equation}
				\label{eq:mono_spin_Omega_n_hypergeo}
				\boxed{
				\Omega_{j^N,n} = \binom{
				N+n-1}{
				n}
				 \ \tensor[_{2j+2}]{F}{_{2j+1}} \left(\left. \myatop{-N, -\frac{n}{2j+1},  \ldots ,- \frac{n-i}{2j+1}, \ldots - \frac{n-2j}{2j+1}}{-\frac{N+n-1}{2j+1}, \ldots ,- \frac{N+n-1-i'}{2j+1}, \ldots - \frac{N+n-1-2j}{2j+1}}  \right\vert 1 \right)} \ .
				\end{equation}
\end{proposition}
\begin{proof}
First, we note that we can rewrite the sum in Eq. \eqref{eq:mono_spin_Omega_n} as $\sum^{\infty}_{s=-\infty} t_s$. Moreover, the term ratio $t_{s+1}/t_s$ is
				\begin{equation}
				\frac{t_{s+1}}{t_s} = \frac{\prod^{2j}_{i=0} \ \left( s - \frac{n-i}{2j+1} \right)}{\prod^{2j}_{i'=0} \ \left( s - \frac{N+n-1-i'}{2j+1} \right)}\frac{s-N}{s+1} \ ,
				\end{equation}
and 
				\begin{equation}
				t_0 = \binom{
				N + n - 1 }{
				n
				} \ ,
				\end{equation}
from which we infer Eq. \eqref{eq:mono_spin_Omega_n_hypergeo}.
\end{proof}
We draw the reader's attention to the fact that both the number of upper and lower parameters of $\tensor[_{2j+2}]{F}{_{2j+1}}$ can be reduced by one each any time the following condition is satisfied:
				\begin{equation}
				\label{eq:hypergeo_reduction_Omega_n}
				i' -i = N-1 \ .
				\end{equation}
Those pair of upper and lower parameters characterized by $i$ and $i'$ satisfying Eq. \eqref{eq:hypergeo_reduction_Omega_n} do not contribute to the generalized hypergeometric function in Eq. \eqref{eq:mono_spin_Omega_n_hypergeo}. Thus, in general, one may make the transformation $\tensor[_{2j+2}]{F}{_{2j+1}} \to \tensor[_{2j+2 - \epsilon}]{F}{_{2j+1-\epsilon}}$, where
				\begin{equation}
				\epsilon= \dim R_{j,N,1} \ , \ \mbox{where } R_{j,N,1} := \left\lbrace \left. (i',i) \right| i' - i = N-1 \ \wedge \ 0 \ \leq i,i' \in \mathbb{N}_0 \leq 2j \right\rbrace
				\end{equation}
and where $\tensor[_{2j+2 - \epsilon}]{F}{_{2j+1-\epsilon}}$ excludes all indexed pair of upper and lower parameters $(i',i) $ of $\tensor[_{2j+2}]{F}{_{2j+1}}$ which are elements of $R_{j,N,1}$. Obviously, when $2j < N- 1$ neither the number of upper nor lower parameters in Eq. \eqref{eq:mono_spin_Omega_n_hypergeo} can be reduced.
\par When it comes to the multiplicities $\{\lambda_{\mathpzc{A},n}\}$, we deduce from Eq. \eqref{eq:generalized_lambda} that for univariate spin systems
				\begin{equation}
				\label{eq:mono_spin_lambda_k}
				\boxed{
				\lambda_{j^N,\kappa}= \sum^{\lfloor{\frac{\kappa}{2j+1}\rfloor}}_{s=0} (-1)^s
				\binom{
				N + \kappa - 2 -  (2j + 1)s }{
				N- 2
				} 
				\binom{
				N }{
				s
				} }\ .
				\end{equation}							
This result is identically the same as Polychronakos and Sfetsos'\cite{art:Polychronakos-2016} (who arrived at their result employing common composition rules of $ SU(2) $ representations), after one performs the transformation according to Eq. \eqref{eq:M_to_n}. 
\begin{proposition}
The summation in Eq. \eqref{eq:mono_spin_lambda_k} yields
				\begin{equation}
				\label{eq:mono_spin_lambda_n_hypergeo}
				\boxed{
				\lambda_{j^N,\kappa} = \binom{
				N+\kappa-2 }{
				\kappa
				} \ \tensor[_{2j+2}]{F}{_{2j+1}} \left(\left. \myatop{-N, -\frac{\kappa}{2j+1},  \ldots ,- \frac{\kappa-i}{2j+1}, \ldots - \frac{\kappa-2j}{2j+1}}{-\frac{N+\kappa-2}{2j+1}, \ldots ,- \frac{N+\kappa-2-i'}{2j+1}, \ldots - \frac{N+\kappa-2-2j}{2j+1}}  \right\vert 1 \right)} \ .
				\end{equation}
\end{proposition}
We leave the proof to the reader. Just as in the case of $\Omega_{j^N,\kappa}$, Eq. \eqref{eq:mono_spin_Omega_n_hypergeo}, the hypergeometric function in Eq. \eqref{eq:mono_spin_lambda_n_hypergeo} can be reduced to  $\tensor[_{2j+2-\epsilon'}]{F}{_{2j+1-\epsilon'}}$ where
				\begin{equation}
				\epsilon'= \dim R_{j,N,2} \ , \ \mbox{where } R_{j,N,2} := \left\lbrace \left. (i',i) \right| i' - i = N-2 \ \wedge  \ 0 \ \leq i,i' \in \mathbb{N}_0 \leq 2j \right\rbrace \ ,
				\end{equation}
on the condition that $2j \geq N-2$.

\par We mention that the generating function $G_{\mathpzc{A}, \Omega}(q)$ for an univariate spin system is similar to that provided in \citep{art:Polychronakos-2016} -- the difference being a constant of proportionality which depends exponentially on $J_0$. In particular, in \citep{art:Polychronakos-2016} the generating function was arrived at by considering the (statistical mechanical) partition function $\mathpzc{Z}_{\mathpzc{A}}$ of a system of noninteracting identical spins coupled with a weak magnetic field of intensity $B$ along the $z-$axis. Assuming the kinetic energy of each spin is negligible compared to the interaction energy with the static magnetic field, then the Hamiltonian of the system is thus $H=-\mu g B\sum_ij_{i,z}$, where $\mu$ and $g$ are the appropriate magneton and $g-$ factor, respectively. 
\par Indeed, if we let $q \to e^{-\beta'}$, then in general,
				\begin{equation}
				\label{eq:Partition_function_G_A}
				\mathpzc{Z}_{\mathpzc{A}} = e^{-\beta'_0} \ G_{\mathpzc{A}, \Omega}\left(e^{-\beta'}\right) 
				\end{equation}
for any IS system interacting with a magnetic field as described above, where $\beta' := \beta \epsilon$, and $\beta'_0 := \beta E_0$; $\epsilon$ is the absolute value of the constant quanta of energy involved in the transition between two consecutive energy levels, and $E_0$ is the ground state energy of the system. In the case of a noninteracting univariate spin system with Hamiltonian $H$ given above, $\epsilon= \mu \lvert g \rvert B$ and $E_0 =-\epsilon J_0$.
\par It is worthwhile to note that Eq. \eqref{eq:Partition_function_G_A} shows that for a noninteracting IS system interacting with a magnetic field under the conditions outlined above, the ground state configuration of the system basically determines its thermodynamics at low temperatures. At high temperatures $(\beta \epsilon \to 0)$, all configurations equally contribute to the thermodynamics of the system and $\mathpzc{Z}_{\mathpzc{A}} \to (2j+1)^N$. If we consider for example an IS ferromagnetic or antiferromagnetic system under the mean field approach, we observe that Eq. \eqref{eq:Partition_function_G_A} still holds, \emph{i.e.} the ground state configuration dominates the partition function. We shall come back to this point in \S \ref{subsec:exclusion_processes}.
\par As shown in \cite{art:Polychronakos-2016}, the CGD of an IS system can be analytically derived considering the partition function $\mathpzc{Z}_{\mathpzc{A}}$ (multiplied by $(1-e^{-\beta'})$) of the system in the presence of a weak static magnetic field. In particular, for such systems $\mathpzc{Z}_{\mathpzc{A}}$ is the generating function for $\{\Omega_n\}$. Given these observations, it might be tempting to conjecture that the CGD of any collection of spins may be determined in a similar fashion. Unfortunately,  that is not the case. This is because for an arbitrary collection of noninteracting spins in the presence of a magnetic field,  $H= -B\sum_i \mu_i g_i j_{i,z}$. Therefore, 
				\begin{equation}
				\label{eq:Z_A_general_case}
				\mathpzc{Z}_{\mathpzc{A}} = e^{-\sum_\alpha \beta'_{0,\alpha}} \prod_\alpha \left(\sum^{2j_\alpha}_{n_\alpha=0} e^{-\beta'_\alpha n_\alpha}\right)^{N_\alpha}
				\end{equation}
where $\alpha$ runs over different spins, $\beta'_{\alpha} := \beta \epsilon_\alpha=\beta \mu_\alpha \lvert g_\alpha \rvert B$, $\beta'_{0,\alpha} := \beta E_{0,\alpha}$, with $E_{0,\alpha} := -\epsilon_\alpha N_\alpha j_\alpha$. Given that in general $\beta'_\alpha \neq \beta'_{\alpha'}$ for $\alpha \neq \alpha'$, there is no way the partition function $\mathpzc{Z}_{\mathpzc{A}}$ in Eq. \eqref{eq:Z_A_general_case} will give the same coefficients $\{ \Omega_n \}$ as $G_{\mathpzc{A},\Omega}(q)$ (see Eq. \eqref{eq:gen_Omega_n_c}), except perhaps one \emph{assumes} $\beta'_\alpha = \beta'_{\alpha'}$ for $\alpha \neq \alpha'$.
The point is: the basic requirement for the partition function to be a good generating function for $\{\Omega_n\}$, thus the precursor of the generating function for the CGD spin multiplicities, is that the quanta energy $\epsilon$ remains invariant for all the spins of the collection. But such an invariance is in general not guaranteed by the physics of the problem. Thus, $\mathpzc{Z}_{\mathpzc{A}}$ is in general a good generating function for $\{\Omega_n\}$ only in the case of a system of identical spins whereas $G_{\mathpzc{A},\Omega}$ knows not such a limitation -- a reminder of the fact that the problem of analytically determining the CGD of an arbitrary collection of spins is an enumerative combinatoric one.
				
\section{Applications and further considerations}\label{sec:applications}
\subsection{A simple illustration}\label{subsec:application_illustration}
Let us consider a spin system of four spin-$1$ and two spin-$1/2$. Thus, $\mathpzc{A} = \left\lbrace\frac{1}{2} , \frac{1}{2} , 1 ,1 ,1 ,1\right\rbrace$. This could be, for instance, the case of a monoradical ion with $6$ nonzero spin nuclei: four of spin-1 and two of spin-$1/2$. The addition of these six angular momenta will yield a series of angular momenta (\emph{i.e.} the Clebsch-Gordan series), which will constitute the set $\mathpzc{E}_{\mathpzc{A}}$ (see Eq. \eqref{eq:multiset_E}). The maximum of $\mathpzc{E}_{\mathpzc{A}}$, $J_0$, according to Eq. \eqref{eq:J_0} is $J_0 = 5$. In this particular case, $\nu_{\mathpzc{A}}=-3$ and so the minimum of the Clebsch-Gordan series (\emph{i.e.} of $\mathpzc{E}_{\mathpzc{A}}$) is $J_m=0$ (see Eqs. \eqref{eq:J_m} and \eqref{eq:upislon_A}). So, the distinct spin angular momenta we get from the addition of the six initial spin angular momenta are 
					\begin{equation}
					J_0=5, \quad J_1=4,\quad  J_2=3, \quad J_3=2,\quad  J_4=1,\quad  J_5=0 \ .
					\end{equation}
We now determine the multiplicity of the various distinct $J_\kappa$ employing the three methods described above.

\subsubsection{The multi-restricted composition method.}
\par The multiplicity of the total angular momenta $\{J_\kappa\}$ can be determined by first determining the first six values of $\{\Omega_n\}$ (see Eqs. \eqref{eq:dim_B(n)} and \eqref{eq:gen_Omega_n}). Employing Eq. \eqref{eq:dim_B(n)}, we may calculate $\Omega_0, \Omega_1 , \ldots , \Omega_5$. For the sake of brevity, we illustrate here the calculation of $\Omega_4$ and $\lambda_4$. First of all, we need to determine the set $P(\mathpzc{A};n=4)$. In other words, we need to determine all the possible ways of writing the integer $4$ as the sum of at most six integers, with the restriction that no more than four parts can be greater than $2$, and $0$ is an admissible part. With this prescription, we find that
					\begin{equation}
					\begin{split}
					P(\mathpzc{A};n=4) = \{A(4)\} & = \{ (1,1,1,1,0,0), (2,1,1,0,0,0), (2,2,0,0,0,0) \} \\
					& = \{ (1,1,1,1), (2,1,1), (2,2) \} \ .
					\end{split}
					\end{equation}
Let us call $(1,1,1,1) := A_x(4)$, $(2,1,1) := A_y(4)$ and $(2,2):= A_z(4)$. Then, $\widetilde{A}_x(4) = \{1\}, \ \widetilde{A}_y(4) = \{1,2\}$ and $\widetilde{A}_z(4) = \{2\}$ and so from Eq. \eqref{eq:dim_B(n)} it follows that
					\begin{equation}
					\label{eq:Omega_4-a}
					\Omega_4 = \binom{
					6 }{
					4
					} + \binom{
					4 }{
					1
					} \binom{
					5 }{
					2
					} + \binom{
					4 }{
					2
					}  = 61
					\end{equation}
where the first, second and third terms are the contributions from $A_x(4), A_y(4)$ and $A_z(4)$, respectively. 
\par According to Eq. \eqref{eq:lambda_recursive}, in order to determine $\lambda_4$ we need $\Omega_3$. For $n=3$ here, 
					\begin{equation}
					P(\mathpzc{A};n=3)  = \{ (1,1,1), (1,2) \} \ .
					\end{equation}
And so from Eq. \eqref{eq:dim_B(n)} we find that
					\begin{equation}
					\Omega_3 = \binom{
					6 }{
					3
					} + \binom{
					4 }{
					1
					} \binom{
					5 }{
					1
					} = 40 \ ,
					\end{equation}
which means $\lambda_4=21$, according to Eq. \eqref{eq:lambda_recursive}. 

\subsubsection{The generalized binomial approach.}
For the same reason as above, we only illustrate here how to calculate $\Omega_4$ and $\lambda_4$ using the generalized binomial method.
\par With $\mathpzc{A}=\{\frac{1}{2}^2, 1^4\}$, we only have two distinct momenta. Let us associate the dummy variable $s_1$ with $j_\alpha = 1/2$ and $s_2$ with $j_\alpha = 1$. Then, from Eq. \eqref{eq:generalized_dim_B(n)} it follows that
					\begin{equation}
					\Omega_4 = \sum_{\substack{2s_1 + 3s_2 \leq 4 \\ 0 \leq s_1 \leq 2 , \ 0  \leq s_2 \leq 4 }} (-1)^{s_1 + s_2} \binom{
				9 - 2s_1-3s_2 }{
				5
				} 
				\binom{
				2 }{
				s_1
				} 
				\binom{
				4 }{
				s_2
				} \ .
					\end{equation}
The only pairs of $s_1$ and $s_2$ which satisfy the condition $2s_1 + 3s_2 \leq 4$ are $(s_1,s_2) \in \{(0,0), (1,0), (2,0), (0,1) \}$, which correspond to four summands:
				\begin{equation}
				\Omega_4 = \binom{
				9 }{
				5
				} - \binom{
				7 }{
				5
				} \binom{
				2 }{
				1
				} + \binom{
				5 }{
				5
				}\binom{
				2 }{
				2
				} - \binom{
				6 }{
				5
				}\binom{
				4 }{
				1
				} = 61 \ .
				\end{equation}
Similarly, in accord with Eq. \eqref{eq:generalized_lambda} we find that
					\begin{equation}
					\begin{split}
					\lambda_4 & = \sum_{\substack{2s_1 + 3s_2 \leq 4 \\ 0 \leq s_1 \leq 2 , \ 0  \leq s_2 \leq 4 }} (-1)^{s_1 + s_2} \binom{
				8 - 2s_1-3s_2 }{
				4
				}
				\binom{
				2 }{
				s_1
				} 
				\binom{
				4 }{
				s_2
				} \ , \\
				& = \binom{
				8 }{
				4
				} -
				\binom{
				6 }{
				4
				} 
				\binom{
				2 }{
				1
				} +
				\binom{
				4 }{
				4
				} 
				\binom{
				2 }{
				2
				} - 
				\binom{
				5 }{
				4
				}
				\binom{
				4 }{
				1
				}	 = 21 \ .							
				\end{split}
					\end{equation}
\subsubsection{The generating function method.} 
\par The numbers $\Omega_n$ are easily determined using the generating function given in Eq. \eqref{eq:gen_Omega_n}. In this case, we have
					\begin{equation}
					G_{\mathpzc{A},\Omega}(q)=(1+q)^2 (1+q+q^2)^4  = \sum^{2J_0}_{n=0} \ \Omega_n \ q^n  \ .
					\end{equation}
But, 
					\begin{equation}
					\label{eq:gen_6}
					(1+q)^2 (1+q+q^2)^4 = q^{10} + 6q^9 + 19q^8+ 40q^7 + 61q^6 + 70q^5 + 61q^4 + 40q^3 + 19q^2 + 6q + 1 \ .
					\end{equation}
From Eq. \eqref{eq:gen_6} we see that $\Omega_4=61$, as calculated earlier according to both the multi-restricted composition and generalized binomial methods.

\par We are now ready to calculate the multiplicity of $J_0, J_1, \ldots , J_m$. From Eqs. \eqref{eq:lambda_recursive} and \eqref{eq:gen_6} we have that
					\begin{equation}
					\lambda_0 = 1 \ , \ \lambda_1 = 5 \ , \ \lambda_2 = 13 \ , \ \lambda_3 = 21 \ , \ \lambda_4 = 21 \ , \ \lambda_{m=5} = 9 . 
					\end{equation}
The same result can be obtained if we make use of the generating function $G_{\mathpzc{A},\lambda}(q)$ in Eq. \eqref{eq:gen_lambda_n}. Indeed, in this case,
					\begin{equation}
					\label{eq:illustration_gen_lambda}
					G_{\mathpzc{A},\lambda}(q) = 1+ 5q + 13q^2+ 21q^3 +21q^4+ 9q^5 - 9q^6 - 21q^7- 21q^8- 13q^9- 5q^{10} - q^{11} \ ,
					\end{equation}
from which we observe that $\lambda_4=21$ as previously determined. Note the absence of sinking terms in Eq. \eqref{eq:illustration_gen_lambda}: in fact, $m=J_0=5$.
\par Adopting a notation similar to that in \cite{art:Zachos-1992}, we may write the Clebsch-Gordan decomposition series from the coupling of $2$ spin-$1/2$ and $4$ spin-$1$ angular momenta as 
					\begin{equation}
					\mathbf{2}^{\otimes 2} \otimes \mathbf{3}^{\otimes 4} = \mathbf{11} \oplus 5 \cdot \mathbf{9} \oplus 13 \cdot \mathbf{7} \oplus 21 \cdot \mathbf{5}\oplus 21 \cdot \mathbf{3} \oplus 9 \cdot \mathbf{1} \ ,
					\end{equation}
where the integers in boldface represent the dimension of a representation; the integers multiplying the boldface integers on the right-hand side are the multiplicities.
The Hilbert space $\mathcal{H}$ of the whole system is of dimension $324$. With the above calculated multiplicities, one can verify that
					\begin{equation}
					\sum^m_{\kappa = 0} \lambda_\kappa \ (2J_\kappa + 1) = 324 \ ,
					\end{equation}					 
as should be expected.
\par In passing, we note that term symbols in atomic physics can be readily determined in similar fashion.
\subsection{Isotropic tensors, Riordan numbers, Catalan numbers and lattice paths}\label{subsec:counting_numbers}
A rank $N$ tensor in a $D$ dimensional Euclidean space, $T_{D,N}$, has the same irreducible representations as the composition of $N$ spin-$\left(\frac{D-1}{2}\right)$ representations in $SU(2)$. The dimension of the basis set of its irreducible components are given by the (positive) coefficients of $G_{\mathpzc{A},\lambda}(q)$; here, $\mathpzc{A}=\{(\frac{D-1}{2})^N\}$. For example, the irreducible components of a rank $10$ tensor in $3-$dimensional space have the same dimensions as the irreducible components of the coupling of $10$ spin-$1$ particles:
					\begin{multline}
					\mathbf{3}^{\otimes 10} = 603 \cdot \mathbf{1} \oplus 1585 \cdot \mathbf{3} \oplus 2025 \cdot \mathbf{5} \oplus 1890 \cdot \mathbf{7} \oplus 1398 \cdot \mathbf{9}\oplus 837 \cdot \mathbf{11}  \\ 
					\oplus 405 \cdot \mathbf{13} \oplus 155 \cdot \mathbf{15} \oplus 45 \cdot \mathbf{17} \oplus 9  \cdot \mathbf{19} \oplus \mathbf{21}  \ .
					\end{multline}
The irreducible component $\mathbf{1}$ is the $S$ state, \emph{i.e.} the totally symmetric part of the tensor. The multiplicities of the various representations are also the characters of the rotation group in that representation; which means that the totally symmetric representation of a rank $10$ tensor in $D=3$ has a basis set of dimension $603$.  These are also the number of independent isotropic tensor isomers present in $T_{3,10}$. 
\par Isotropic tensors are crucial when determining rotational averages of observables, and they are particularly useful in essentially all studies related to matter-radiation interaction\cite{art:Andrews-1977}. To correctly perform rotational averages, it is of fundamental importance to know beforehand how many linearly independent isotropic isomers there are. As we have shown above, one can use CGD to determine this number by making use of the relations derived above. We also mention that the number of independent isotropic isomers in $D=3$ of rank $n=0,1,2,\ldots$ are collectively called Motzkin sum \cite{book:Paolucci-2016,*link:Weisstein} or Riordan\cite{art:Bernhart-1999} numbers and their generating function is\cite{link:Weisstein}
					\begin{equation}
					\begin{split}
					G_{3D}(q)& =\frac{1}{2q}\left(1-\sqrt{\frac{1-3q}{1+q}} \right)\\
					& = 1 + q^2 + q^3 + 3q^4 + 6q^5 + 15q^6 + 36q^7 + 91q^8 + 232q^9 + \ldots 
					\end{split}
					\end{equation}					  
Following similar arguments, a rank $N$ tensor in $D=2$ has the same irreducible representations as a multispinor of rank $N$. We thus understand from Eq. \eqref{eq:J_m} that there will be no isotropic tensors when $N$ is odd since in that case $J_m=1/2$. In fact, any tensor of rank odd (even) $N$ in $D=2$ is fermionic (bosonic). In particular, the number of independent isotropic tensors in $T_{2,N=2n}, \ n \in \mathbb{N}_0$ is given by $\lambda_{(1/2)^N,N/2}$, \emph{i.e.} the spin multiplicity  of the $N/2-$th distinct element of the associated multiset $\mathpzc{E}_{(1/2)^N}$ obtained after coupling $N$ spin-(1/2)s. We thus infer from Eqs. \eqref{eq:dim_B(n)} or \eqref{eq:spin_1/2_Omega_limit} and \eqref{eq:lambda_recursive} that
					\begin{equation}
					\begin{split}
					\lambda_{(1/2)^N,N/2} & = \binom{
					N }{
					N/2
					} - \binom{
					N }{
					N/2 - 1
					} = \frac{1}{N/2+1} \binom{
					N }{
					N/2
					} \\
					& := C_{N/2}
					\end{split}
					\end{equation}
where we have recognized the integers $\{\lambda_{(1/2)^N,N/2}\}$ as being the well-known Catalan numbers, $\{C_{N/2}\}$. More appropriately, the $\{C_{N/2}\}$ are the so-called aerated Catalan numbers\cite{art:Donaghey-1977} since $C_{N/2}=\lambda_{(1/2)^N,N/2}=0$ for odd $N$. The generating function $G_{2D}(q)$ for the integers $\{\lambda_{(1/2)^N,N/2}\}$ is of the form\cite{art:Donaghey-1977}
					\begin{equation}
					\begin{split}
					G_{2D}(q)& =\frac{1}{2q^2}\left(1-\sqrt{1-4q^2} \right)\\
					& = \sum^{\infty}_{N=0} \ \lambda_{(1/2)^N,N/2} \ q^{N} \\
					& = 1 + q^2 + 2 q^4 + 5 q^6 + 14 q^{8} + 42 q^{10} + 132 q^{12} + \ldots 
					\end{split}
					\end{equation}
where we notice the appearance of only even terms (corresponding to bosonic states). 
From Eqs. \eqref{eq:lambda_recursive} and \eqref{eq:spin_1/2_Omega_limit} we have that
					\begin{equation}
					\lambda_{(1/2)^N,\kappa} = \frac{N+1-2\kappa}{N+1-\kappa} \binom{
					N }{
					\kappa
					} \ , \ \qquad \kappa \in \{0,1,\ldots , m=N/2 \} \ . 
					\end{equation}
Drawing on the results obtained with the generalized binomial method (Eq. \eqref{eq:mono_spin_lambda_n_hypergeo}), we end up with the following identities:
					\begin{subequations}
					\begin{align}
					 \lambda_{(1/2)^N,\kappa} & = 
				 	\binom{
					N+\kappa-2 }{
					\kappa
					} \ \tensor[_{3}]{F}{_{2}} \left(\left. \myatop{-N, -\frac{\kappa}{2},  - \frac{\kappa-1}{2}}{-\frac{N+\kappa-2}{2}, -\frac{N+\kappa-3}{2}}  \right\vert 1 \right) \\
					C_\nu & = \binom{
					3\nu-2 }{
					\nu
					} \ \tensor[_{3}]{F}{_{2}} \left(\left. \myatop{-2\nu, -\frac{\nu}{2},  -\frac{\nu-1}{2}}{-\frac{3\nu-2}{2}, -\frac{3\nu-3}{2}}  \right\vert 1 \right) \\ 
					R_\nu & = \binom{
					2\nu-2 }{
					\nu
					} \ \tensor[_{4}]{F}{_{3}} \left(\left. \myatop{-\nu, -\frac{\nu}{3},  -\frac{\nu-1}{3}, - \frac{\nu-2}{3}}{-\frac{2\nu-2}{3},-\frac{2\nu-3}{3}, -\frac{2\nu-4}{3}}  \right\vert 1 \right) \ .
					\end{align}
					\end{subequations}			
where $\nu \in \{0, 1, 2,3 , \ldots\}$, and $\{C_\nu\}$ and $\{R_\nu\}$ are the Catalan and Riordan numbers, respectively. The easiness with which these identities are derived from Eq. \eqref{eq:mono_spin_lambda_n_hypergeo} is worth noting. 
\par Moreover, the popping up of Catalan and Riordan numbers in these limit cases is very telling. They strike at a deeper connection between enumerative combinatorics and CGD. In particular, given that the Catalan and Riordan numbers are related to some counting problems in lattice paths (or random walks), it is only fair to ask if every CGD problem can be re-interpreted as a lattice path problem. We call this the "CGD-lattice path duality" problem. A quick analysis of the issue seems to weigh in favor of the affirmative, but a more formal prove need to be given. For example, a very important characteristic of the lattice paths counted by Catalan triangle is that  at each step one can move up $(\Delta =+1)$ or down $(\Delta = -1)$, while with paths counted by the Motzkin numbers one can stay put $(\Delta = 0)$ besides the up and down moves\citep{art:Donaghey-1977}. These restricted variations of one's position $(\Delta = 0, \pm 1)$ are reminiscent of selection rules for spin transitions. And the steps may be interpreted as multiples of the transition timescale.  So far, the random walk connection has been discussed in the literature only in the case of $ N $ spin-$ j $s \citep{art:Polychronakos-2016} (see also \citep[sec. 3.2]{art:Curtright-2017}). Extending this interpretation to the generalized multivariate $\mathpzc{A}$ may set the ground for a formal way of mapping lattice gas models to spin dynamics, which might prove to be computationally advantageous (especially in EPR and NMR simulations). Movassagh and Shor's\citep{art:Movassagh-2014} recent application of lattice paths, among other mathematical techniques, to prove the violation of the area law in some $D=1$ models is particularly interesting and could offer some insights as to how to achieve the lattice gas-spin dynamics mapping in general.
\subsection{CGD and symmetric exclusion process on graphs.}\label{subsec:exclusion_processes}
An even more interesting application of our results has to do with exclusion processes (see \citep{art:Mendonca-2013} and references therein). In \cite{art:Mendonca-2013}, Mendonça considering a simple symmetric exclusion process (SSEP) on a complete graph (with characteristics defined in the article), showed that its infinitesimal generator $\mathcal{H}$ is equivalent to the Hamiltonian of the isotropic Heisenberg spin-$1/2$ quantum ferromagnet on the complete graph. The conservation of particles in the process implied that $\mathcal{H}$ can be blocked-diagonalized, with each invariant subspace conserving $n$ particles. The peculiar characteristic of the SSEP model is that each vertex of the complete graph can accommodate not more than a particle at a time. We can thus assign to each vertex the number $0$ (when it is empty) or $1$ (when occupied). Here, these values are literally counting the number of particles in that vertex, and the fact that we can associate or $0$ or $1$ to each vertex is an indication that each vertex in this model is acting like a spin-1/2 in $SU(2)$. The degeneracy of the eigenstate characterized by a total number of particles $n$ is simply $\Omega_{(1/2)^ N,n}$ -- which makes perfect sense if we recall how the expression for $\Omega_{\mathpzc{A},n}$ was derived under the multi-restricted composition method in Sec. \ref{sec:momentum_inv_sub_dim}. 
\par Our results can be used to analyze the degeneracy of more complicated scenarios. One may consider for example a complete graph of $N$ vertices, whereby the vertex $i$ can accommodate at most a given finite number $n_i$ of particles (or assume a finite number of states which can be ordered). Say we represent the maximal occupation of the vertices by the multiset $\{n_1^{d_1}, \ldots , n_\sigma^{d_\sigma}\}$, where $d_\alpha$ is the number of vertices which can accommodate at most $n_\alpha$ particles and $\sum^\sigma_{\alpha=1} d_\alpha = N$. Then the degeneracy of the invariant subspace characterized by $n$ particles is $\Omega_{\mathpzc{A},n}$ (Eq. \eqref{eq:generalized_dim_B(n)}), where $\mathpzc{A}=\{(\frac{n_1}{2})^{d_1}, \ldots , (\frac{n_\sigma}{2})^{d_\sigma}\}$. Here, we should expect the generator of the process to be equivalent to the Hamiltonian of a mixed quantum spin model. 
\par In the special case whereby each vertex can be occupied by any number of particles (known in the literature as the "zero range process"), we lose exclusion and $\mathpzc{A}=\{\infty^N\}$ (\emph{i.e.} the process becomes equivalent to the composition of $N$ spin-$\infty$). We thus end up with vertices which act like bosonic eigenstates. And the generating function for $\Omega_{\infty^N, n}$, $G_{\infty^N, \Omega}(q)$, is
						\begin{equation}
						\label{eq:gen_infty_Omega}
					G_{\infty^N, \Omega}(q)=(1 + q + q^2 + q^3 + \ldots )^N = \sum^{\infty}_{n=0}  \ \Omega_{\infty^N, n} \ q^n \ .
					\end{equation}
But $(1 + q + q^2 + q^3 + \ldots )^N= (1-q)^{-N}$, from which we derive that
					\begin{equation}
					\label{eq:mono_Omega_infinity}
					\Omega_{\infty^N, n} = \binom{
					N + n - 1 }{
					n
					}
					\end{equation}
and,
					\begin{equation}
					\label{eq:mono_lambda_infinity}
					\lambda_{\infty^N, \kappa} = \binom{
					N + \kappa - 2 }{
					\kappa
					} \ .
					\end{equation}
$\lambda_{\infty^N, n}$ is therefore the number of ways one can select $\kappa$ objects from $N-1$ distinct objects, with no constraint on the number of times an object may appear.   Comparing Eqs. \eqref{eq:mono_Omega_infinity} and \eqref{eq:mono_lambda_infinity} with Eqs. \eqref{eq:mono_spin_Omega_n_hypergeo} and \eqref{eq:mono_spin_lambda_n_hypergeo}, respectively, we see that both hypergeometric functions in the latter pair satisfy the limit: $\lim_{j \to \infty} \ \tensor[_{2j+2}]{F}{_{2j+1}}\left(\left. \myatop{\ldots}{\ldots} \right| 1 \right) = 1 $. 
\par The notion of spin-$\infty$ limit might seem too exotic to embrace, let alone the magnetism of such systems. But we recall that magnetism in the limit of spin-$\infty$ under the isotropic Heisenberg model has been discussed by Fisher\cite{art:Fisher-1964}, who -- upon deriving the partition function of the system in the zero field limit -- was able to derive other thermodynamic quantities like the free energy, specific heat, magnetic susceptibility and correlation functions. The operators $j_{x}, j_{y}$ and $j_{z}$ do commute in the spin-$\infty$ limit, reducing therefore the model to that of a classical Heisenberg one. As remarked in \cite{art:Fisher-1964}, the thermodynamics of spin-$\infty$ seems impossible to describe analytically in the presence of a finite magnetic field using the normal textbook methods. We claim that this infinite-spin problem can be solved employing the EC method discussed above.
\par To this end, let us consider a finite number $N$ of noninteracting spin-$\infty$s in the presence of a static magnetic field of magnitude $B$, again assuming their kinetic energy is negligible with respect to the interaction energy with the field. For such a spin system, $E_0 \to -\infty$. We are dealing here with a system of identical spins in the presence of a magnetic field and so Eq. \eqref{eq:Partition_function_G_A} applies; namely, the partition function of the system is
				\begin{equation}
				\begin{split}
				\mathpzc{Z}_{\infty^N} & = e^{-\beta_0} \ G_{\infty^N, \Omega} \left(e^{-\beta'}\right) \\
				& = e^{-\beta_0} \ \sum^{\infty}_{n=0} \binom{
					N + n - 1 }{
					n
					} e^{-\beta'n} \\
					& = e^{-(\beta'_0-\frac{1}{2}\beta'N)} \left[2 \sinh \left(\frac{1}{2}\beta'\right) \right]^{-N}
				\end{split}
				\end{equation}				 
where we have made use of Eq.s \eqref{eq:gen_infty_Omega} and \eqref{eq:mono_Omega_infinity}. It is interesting to note that if we \emph{assume} the ground state energy of the system to be finite, then $\mathpzc{Z}_{\infty^N}$ becomes none other but the quantum mechanical partition function of a system of $N$ independent harmonic oscillators whenever $\beta'_0 = \frac{1}{2} \beta' N$; the correspondence becomes even clearer if we express the quanta $\epsilon$ in units of $\hbar$, \emph{i.e.} $\epsilon = \hbar \omega$. This formal analogy allows us to readily obtain analytical expressions for the thermodynamic quantities for $N$ noninteracting spin-$\infty$ in the presence of a magnetic field by just looking at their corresponding expression for the quantum mechanical treatment of $N$ independent harmonic oscillators. For example, the Helmholtz free energy $(F)$, internal energy $(U)$, entropy $(S)$, and specific heats at constant volume $(C_V)$ and pressure $(C_P)$ for $N$ noninteracting spin-$\infty$ in a magnetic field $B$ are given by the following expressions:
					\begin{subequations}
					\label{eq:F_U_S_C_P}
					\begin{align}
					F - E_0 & =  k_BT N \ln (1-e^{-\beta'})	\\
					U - E_0 & =   N\beta^{-1}\frac{\beta'}{e^{\beta'}-1} \\
					\frac{S}{k_B} & = N \left[\frac{ \beta'}{e^{\beta'}-1} - \ln \left(1-e^{-\beta'} \right) \right] \\
					C_P & = C_V = Nk_B \beta'^2 \ \frac{e^{\beta'}}{\left(e^{\beta'}-1 \right)^2} \ .
					\end{align}
					\end{subequations}
In the above, we see that even in the limit $E_0 \to -\infty$ we can define $F$ and $U$ with respect to the ground state energy and get finite quantities for the corresponding rescaled energies at low temperatures. Indeed, from Eq. \eqref{eq:F_U_S_C_P} it happens that the rescaled energies and the other thermodynamic quantities all depend on $\beta'$ and $N$, which are both finite -- therefore, the thermodynamics of the system can be considered to be well-defined. In particular, the dependence of the thermodynamics on $\beta'$, \emph{i.e.} the ratio between the  transition energy between two consecutive levels and the mean energy per spin, is very interesting and one can easily analyze the limit cases $\beta' \gg 1$ and $\beta' \ll 1$. The magnetization and the spontaneous magnetization of the system are both easily shown to tend to  $(-\infty)$ -- which is intuitively right since $E_0 \to -\infty$. The susceptibility of the system is found to tend to $(+\infty)$ at high temperature and $0$ at low temperature. Interestingly, the derivative of the magnetization  with respect to the field $B$, for very small $B$ and $\beta' \ll 1$, is $\sim \frac{N}{V}\frac{k_B T}{ B^2}$ -- clearly a violation of Curie's law, namely being $\varpropto T $ instead of $\varpropto 1/T$. These are unfamiliar magnetic properties.

\subsection{Symmetric and Antisymmetric compositions of a system of identical spins}\label{subsec:symmteric_antisymmetric}
The discussion of symmetric exclusion processes (SEP) above allows us to undertake yet another important application of the EC approach to CGD, \emph{viz.} the symmetric and antisymmetric composition of spins. For starters, it is clear that we can only speak of these type of spin compositions when we are dealing with identical spins. 
\par Reinterpreted in the language of symmetric exclusion processes, determining the set $\{\Omega_n\}$ of a collection of $N$ $j-$spins is the same as determining how many ways one can distribute $0\leq n \leq 2J_0$ particles on a complete graph of $N$ vertices, in which each vertex can accommodate not more than $2j$ particles. If we should translate the symmetric exclusion process picture back into $SU(2)$ representation, we have that the $N$ vertices represent the spins and the particles are the Holstein-Primakoff bosons\cite{art:Holstein-1940} -- whose occupation number indicate a particular orientation of the spin. In table \ref{tab:SEP_SU2_dictionary} we compile a short dictionary which summarizes these translations.
\begin{table}[h!]
\centering
\begin{tabular}{|c|c|}
\hline
SEP language & $SU(2)$ language \\
\hline
Vertex & Spin\\
Particle & Holstein-Primakoff boson \\
Vertex of maximum occupancy $2j$ & Spin-$j$\\
Occupancy of a vertex & Spin orientation\\
Complete graph & Interaction of each spin with all others\\
\hline
\end{tabular}
\caption{Dictionary for translating SEPs into $SU(2)$ language, and vice versa.}
\label{tab:SEP_SU2_dictionary}
\end{table}

Thinking now in SEP terms, we see that the key characteristic of the CGD counting problem described above is that both the vertices and the particles are assumed \emph{distinguishable}, which translated into $SU(2)$ language signifies that we are considering all the $j-$spins as distinguishable, along with their orientations. It is indisputable that the various orientations of a $j-$spin are distinguishable, but the spins in an IS system are indeed \emph{indistinguishable}, as required by quantum mechanics.  Recognizing the indistinguishableness of the spins is a springboard for achieving a symmetric (bosonic, $-$) or antisymmetric (fermionic, $+$) composition of an IS system. 
With this in mind, we now consider these representations separately.

\subsubsection{Antisymmetric composition of an IS system}
The way to go about this problem is very similar to how we obtained the CGD of a collection of spins, considering them as distinguishable; namely, we first find the generating function for $\Omega_{\mathpzc{A},n}$, \emph{i.e.} $G_{\mathpzc{A},\Omega}(q)$, then multiply it by $(1-q)$ to get the generating function for $\lambda_{\mathpzc{A},n}$, $G_{\mathpzc{A},\lambda}(q)$.
\par Let $\mathpzc{A}=\left\lbrace j^N\right\rbrace$. Besides the indistinguishableness of the spins, the other prerequisite to satisfy in order to obtain the generating function for $\{\Omega^+_n \}$ and $\{\lambda^+_n \}$ of the antisymmetric composition of an IS system is to impose the condition that no more than one spin can have the same orientation (\emph{antisymmetric composition constraint}). In SEP terms, we are imposing the constraint that two vertices of the same complete graph cannot have the same occupation number of particles. In the following, we shall refrain from switching directly to the integer representation of the orientations of a spin  right from the beginning. In lieu, we shall do the switching later on in our discussion. 
\par The generating function for $\{\Omega^+_n\}$ demands that the number of spins (or vertices) be known and fixed. Since $\Omega^+_n$ counts the number of ways of obtaining the total $z-$eigenvalue $M$ which corresponds to $n$ in the integer representation (given a fixed number of spins), in imposing the antisymmetric constraint we also need to keep record of the number of spins contributing to $M$. First of all, an orientation may contribute only once or not at all to $M$, according to the antisymmetric constraint. In addition, it is obvious that an orientation $m$ ($-j \leq m \leq j $) contributing to $M$ will increase the total $z-$eigenvalue by $m$ and increase the number of spins by $1$. From these, we conclude that the generating function for $\{\Omega^+_{M} \}$, $\mathscr{G}^+_{j, \Omega}(x)$, must be
				\begin{equation}
				\label{eq:antisym_gen_Omega}
				\mathscr{G}^+_{j, \Omega}(x) = \prod^{j}_{m=-j} \left( 1+ax^m \right) = \sum_{M,N} \Omega^+_{j^N,M}\  x^M a^N
				\end{equation}
where $\Omega^+_{j^N,M}$ is the number of ways of obtaining the total $z-$component $M$ for an IS system of $N$ spin$-j$s upon imposing the antisymmetric constraint. Thus, while $x$ keeps track of the total $z-$component, $a$ counts the spins. To verify that the antisymmetric constraint is satisfied by $\mathscr{G}^+_{j, \Omega}(x)$, we may expand the product in Eq. \eqref{eq:antisym_gen_Omega} using the binomial theorem and we get
				\begin{equation}
				\prod^{j}_{m=-j} \left( 1+ax^m \right) = \sum_{N} \left[ \sum_{\substack{s_0, s_1, \ldots , s_{2j}\geq 0 \\ s_0 + s_1 + \ldots + s_{2j} =N}} \binom{1}{s_0} \binom{1}{s_1} \ldots \binom{1}{s_{2j}} x^{s_0\cdot j + s_1 \cdot (j-1) + \ldots + s_{2j}(-j)} \right] a^N
				\end{equation}
from which it is crystal clear that no spin orientation can contribute more than once to $M$. 
\par If we now switch to the integer representation of the orientations, then $ax^m \to Aq^k$ if $m=j-k \ (0 \leq k \leq 2j)$, where $A := ax^j$ and $q := x^{-1}$. Eq. \eqref{eq:antisym_gen_Omega} may now be rewritten as
				\begin{subequations}
				\begin{align}
				\label{eq:antisym_gen_Omega_2}
				\mathscr{G}^+_{j, \Omega}(q) & = \prod^{2j}_{k=0} \left( 1+Aq^k \right) = \sum_{N,n} \Omega^+_{j^N,n}\  q^n A^N \\
				& = (1 + A)^{2j+1}_q \label{eq:antisym_gen_Omega_2_b}
				\end{align}
				\end{subequations}		
where in Eq. \ref{eq:antisym_gen_Omega_2_b} we have recognized 	$\mathscr{G}^+_{j, \Omega}(q)$ as being the $q-$analogue of $(y+A)^{2j+1}\vert_{y=1}$. Here again, $A$ counts the number of spins.
\par Applying now Gauss' binomial formula\cite{book:Kac-2001}, it follows then from Eq. \eqref{eq:antisym_gen_Omega_2_b} that
				\begin{equation}
				\label{eq:antisym_gen_Omega_3}
				\mathscr{G}^+_{j, \Omega}(q) = \sum^{2j+1}_{N=0} {2j+1 \brack N}_q q^{\binom{N}{2}} A^N = 	 \sum^{2j+1}_{N=0} G^+_{j^N,\Omega}(q) \ A^N		 
				\end{equation}
where ${a \brack b}_q$ is the $q-$binomial coefficient, defined as\cite{book:Kac-2001}
				\begin{equation}
				{\ a  \ \brack \ b \ }_q := \frac{\left[ a \right]_q!}{\left[ a-b\right]_q! \left[ b \right]_q!} = {\ a  \ \brack \ a-b \ }_q \qquad (a,b \in \mathbb{N}_0)
				\end{equation}
and where $\left[a \right]_q!$, the $q-$analogue of $a!$, is defined as
				\begin{equation}
				\left[ a \right]_q! := \begin{cases}
				1  & \mbox{if } a=0\\
				\left[a \right]_q\cdot \left[a-1 \right]_q \cdots \left[1 \right]_q & \mbox{if } a=1,2,3,\ldots
				\end{cases} \ .
				\end{equation}
$G^+_{j^N,\Omega}(q)$ is the  generating function for $\left\lbrace \Omega^+_{j^N,n} \right\rbrace$, and from Eq. \eqref{eq:antisym_gen_Omega_3} we have
				\begin{equation}
				\label{eq:antisymmetric_G_PLUS_int}
				\begin{split}
				G^+_{j^N,\Omega}(q) & := {2j+1 \brack N}_q q^{\binom{N}{2}} = q^{\binom{N}{2}} \prod^{N-1}_{k=0} \frac{\left[2j+2-N+k \right]_q}{\left[1+k \right]_q} \\
				& = \sum^{N(4j+1-N)/2}_{n=\binom{N}{2}} \Omega^+_{j^N,n} \ q^n \ .
				\end{split}
				\end{equation}
				
Resorting to Theorem \ref{thm:q_binomial_ps_explicit} (see Appendix), we conclude that
				\begin{equation}
				\label{eq:antisymmetric_Omega_PLUS_a}
				\Omega^+_{j^N,n} = p(2j+1-N,N,f(n,N)) = \sum^{f(n,N)}_{\nu=0} \phi^{2j+1,N}_{\nu,f(n,N)} \ , \qquad \ f(n,N) :=  n-\binom{N}{2}
				\end{equation}
where $p(n,m,k)$ is the number of distinct partitions of $k$ into at most $m$ parts, each not greater than $n$. For example, $p(3,4,5)= 4$ since the number of partitions of $5$ with at most four parts, each not greater than $3$ are: $(3,2), (3,1,1),(2,2,1), (2,1,1,1)$. $\phi^{n,m}_{\nu,k}$, on the other hand, is the number of distinct partitions of the integer $k$ into \emph{exactly} $\nu$ parts, \emph{i.e.} $m_1 + m_2 + \ldots + m_\nu = k$, where each part is at most $\nu(n-m)$ and $m_1 \geq m_2 \geq \ldots \geq m_\nu$. For example, $\phi^{7,4}_{3,5}=2$ since $(3,1,1)$ and $(2,2,1)$ are the only partitions of $5$ with exactly three parts which satisfy the prescription above. In the Appendix, we derive a general analytical formula (Eq. \eqref{eq:phi_a_bs}) for the $\left\lbrace \phi^{a,b}_{\nu, k} \right\rbrace$ in terms of the Heaviside step-function defined in Eq. \eqref{eq:Heaviside}. Like $G_{\mathpzc{A},\Omega}(q)$, $G^+_{j^N,\Omega}(q)$ is reciprocal and unimodal.
\par Let $ \lambda^+_{j^N,\kappa} $ denote the multiplicity of the $\kappa-$th resulting spin (in the integer representation) after the antisymmetric composition of $N$ $j-$spins. The generating function for the $\left\lbrace \lambda^+_{j^N,\kappa} \right\rbrace$, $G^+_{j^N,\lambda}(q)$, is obtained from $G^+_{j^N,\Omega}(q)$ analogously to Eq. \eqref{eq:gen_lambda_n}, namely,
				\begin{equation}
				G^+_{j^N,\lambda}(q) = (1-q)\ G^+_{j^N,\Omega}(q) = \sum^{1+N(4j+1-N)/2}_{\kappa=\binom{N}{2}} \lambda^+_{j^N,\kappa} \ q^\kappa \ .
				\end{equation}				  
Certainly,
				\begin{equation}
				\label{eq:antisymmetric_lambda_a}
				\lambda^+_{j^N,\kappa} =p(2j+1-N,N,f(\kappa,N)) - p(2j+1-N,N,f(\kappa -1 ,N))  \ .
				\end{equation}
Owing to the reciprocity of $G^+_{j^N,\lambda}(q)$, it is easy to prove that
				\begin{equation}
				\lambda^+_{j^N,\kappa} = -\lambda^+_{j^N,2J_0+1-\kappa}
				\end{equation}	
(recall $J_\kappa=J_0-\kappa=jN-\kappa$). The unimodality of $G^+_{j^N,\Omega}(q)$ implies that $G^+_{j^N,\lambda}(q)$ may also present sinking terms analogously to $G_{\mathpzc{A},\lambda}(q)$. 
\par The multiset $\mathpzc{E}_{j^N}^+$, whose elements are the resulting spins of the antisymmetric composition, is simply
				\begin{equation}
				\mathpzc{E}_{j^N}^+ = \left\lbrace  J_\kappa ^{\lambda^+_{j^N,\kappa}} \left| \lambda^+_{j^N,\kappa} > 0 \right. \right\rbrace \ 
				\end{equation}
and $\max \mathpzc{E}_{j^N}^+ = N(2j+1-N)/2$, whose multiplicity is always $1$. In other words, if we happen to perform an antisymmetric composition of $N (\leq 2j+1)$ spin-$j$ IS system, the largest spin to be generated from such a composition is $J_{\binom{N}{2}}= N(2j+1-N)/2$.
\par We draw the reader's attention to the fact that one can represent $G^+_{j^N,\Omega}$ and $G^+_{j^N,\lambda}$ in real spin projections by making the transformation $q \to x^{-1}$ and multiplying $G^+_{j^N,\Omega}$ and $G^+_{j^N,\lambda}$ by $x^{J_0}$ (which derives from $A^N$, see Eq. \eqref{eq:antisym_gen_Omega_2}). It goes on without saying that working in the integer representation of the spin projections greatly simplify the derivations.
\par Interestingly, in the spin-$\infty$ limit, we derive from Eq. \eqref{eq:antisym_gen_Omega_2_b} that
				\begin{equation}
				\label{eq:antisymm_infty_limit_1}
				\lim_{j \to \infty} \mathscr{G}^+_{j, \Omega}(q) := \mathscr{G}^+_{\infty, \Omega}(q) = \left(1 + A \right)^\infty_q = E^{A/(1-q)}_q
				\end{equation}
where $E^{y}_q$ (Euler's first identity) is one of the $q-$analogues of the exponential function $e^x$, defined as\cite{book:Kac-2001},
				\begin{equation}
				E^{y}_q := \sum^{\infty}_{k=0} q^{\binom{k}{2}} \frac{y^k}{ \left[ k \right]_q! } = \sum^{\infty}_{k=0} q^{\binom{k}{2}} \frac{(1-q)^k y^k}{ (1-q)(1-q^2)\cdots (1-q^k) } \ .
				\end{equation}
The spin$-\infty$ limit of the generating function  $G^+_{j^N,\Omega}(q)$, $G^+_{\infty^N,\Omega}(q)$, readily follows from Eq. \eqref{eq:antisymm_infty_limit_1} (or even Eq. \eqref{eq:antisymmetric_G_PLUS_int}), \emph{viz.}
				\begin{equation}
				\label{eq:antisymm_infty_limit_2}
				G^+_{\infty^N,\Omega}(q)  = \frac{q^{\binom{N}{2}}}{(1-q)(1-q^2)\ldots (1-q^N)} = \sum^\infty_{n=\binom{N}{2}} \Omega^+_{\infty^N,n} \ q^n \ .
				\end{equation}
After a close examination of Eq. \eqref{eq:antisymm_infty_limit_2} we see that
				\begin{equation}
				\label{eq:antisymm_infty_limit_3}
				\Omega^+_{\infty^N,n} = p_N \left( f(n,N) \right)
				\end{equation}
where $p_N(k)$ is the number of partitions of $k$ into at most $N$ parts. Eq. \eqref{eq:antisymm_infty_limit_3} could have also been derived from Eq. \eqref{eq:antisymmetric_Omega_PLUS_a} by just taking the limit $j \to \infty$. Furthermore, the relation for spin multiplicities $\left\lbrace \lambda^+_{\infty^N} \right\rbrace$ of the multiset $\mathpzc{E}^+_{\infty^N}$ also follows from Eq. \eqref{eq:antisymmetric_lambda_a}:
				\begin{equation}
				 \lambda^+_{\infty^N} = p_N \left( f(n,N) \right) - p_N \left( f(n-1,N) \right) \ .
				\end{equation}

\subsubsection{Symmetric composition of identical spins}
Unlike the antisymmetric composition, two or more identical spins can have the same orientation in the symmetric composition (\emph{symmetric constraint}). Translated into SEP language, we are allowing two or more vertices of the complete graph to have the same occupation number of particles. The indistinguishableness of the spins is crucial here as well.
\par Again, let $\mathpzc{A}=\left\lbrace j^N \right\rbrace$. We shall represent the spins resulting from the symmetric composition by the multiset $\mathpzc{E}^-_{j^N}$. Let $\Omega^-_{j^N,M}$ be the number of elements of $\mathpzc{E}^-_{j^N}$ admitting the total spin orientation $M$ as an admissible one. Since the contribution of an orientation $m\ (-j \leq m \leq j)$ to $M$ is simply $m$ and according to the symmetric constraint any number of spins can have the same orientation, the generating function for the integers $\left\lbrace \Omega^-_{M} \right\rbrace$, $\mathscr{G}^-_{j,\Omega}$, is nothing but of the following form
				\begin{equation}
				\label{eq:sym_gen_Omega}
				\begin{split}
				\mathscr{G}^-_{j,\Omega}(x) & = \left(1+ax^j+a^2x^{2j}+\ldots \right) \left(1 + ax^{j-1} + a^2x^{2(j-1)}+\ldots \right)\cdots \left(1+ax^{-j}+a^2x^{-2j}+\ldots \right) \\
				& = \prod^j_{m=-j} \frac{1}{1-ax^m} = \sum_{M,N} \Omega^-_{j^N,M} \ x^M a^N \ ,
				\end{split}
				\end{equation}	%
where we have introduced the variable $a$ to count the spins. The expansion of the product in Eq. \eqref{eq:sym_gen_Omega} using, again, the binomial theorem returns the expression
				\begin{equation}
				\prod^{j}_{m=-j} \frac{1}{\left( 1-ax^m \right)} = \sum_{N} \left[ \sum_{\substack{s_0, s_1, \ldots , s_{2j}\geq 0 \\ s_0 + s_1 + \ldots + s_{2j} =N}} \binom{s_0}{s_0} \binom{s_1}{s_1} \ldots \binom{s_{2j}}{s_{2j}} x^{s_0\cdot j + s_1 \cdot (j-1) + \ldots + s_{2j}(-j)} \right] a^N
				\end{equation}
which certifies a successful implementation of the symmetric constraint.
\par Switching now to the integer representation of the spin orientations as we did earlier, \emph{i.e.} $ax^m \to Aq^k$, Eq. \eqref{eq:sym_gen_Omega} becomes
				\begin{equation}
				\label{eq:symmetric_GENERAL_G_Omega}
				\mathscr{G}^-_{j,\Omega}(q) = \prod^{2j}_{k=0} \frac{1}{1-Aq^k} = \frac{1}{\left(1-A \right)^{2j+1}_q} \ .
				\end{equation}
Applying now Heine's binomial formula\cite{book:Kac-2001}, we obtain
				\begin{equation}
				\mathscr{G}^-_{j,\Omega}(q) = \sum^{\infty}_{N=0} { 2j+N \brack N }_q   A^N = \sum^{\infty}_{N=0} G^-_{j^N,\Omega}(q) \   A^N \ ,
				\end{equation}		
where, of course,
				\begin{equation}
				G^-_{j^N,\Omega}(q) = { 2j+N \brack N }_q = \sum^{2J_0}_{n=0} \Omega^-_{j^N,n} \ q^n \ .
				\end{equation}													
$\Omega^-_{j^N,n}$ is thus the number of Young diagrams for the partition of $n$ which can fit inside a $2j \times N$ rectangle (see Appendix). Indeed, it follows from Theorem \ref{thm:q_binomial_ps_explicit} that,
				\begin{equation}
				\Omega^-_{j^N,n} = p(2j,N,n) = \sum^n_{\nu=0} \phi^{2j+N,N}_{\nu,n} \ .
				\end{equation}
The generating function, $G^-_{j^N,\lambda}(q)$, for the multiplicities $ \left\lbrace \lambda^-_{j^N,\kappa} \right\rbrace$ of the distinct elements of $\mathpzc{E}^-_{j^N}$, is as usual given by the relation
				\begin{equation}
				\label{eq:symm_G_j_N_a}
				G^-_{j^N,\lambda}(q) = (1-q)\ G^-_{j^N,\Omega}(q) = \sum^{2J_0+1}_{\kappa = 0} \lambda^-_{j^N,\kappa} \ q^\kappa \ .
				\end{equation}
We also have that
				\begin{equation}
				\lambda^-_{j^N,\kappa} = p(2j,N,\kappa) - p(2j,N,\kappa-1) \ ,
				\end{equation}
and again,
				\begin{equation}
				\lambda^-_{j^N,\kappa} = -\lambda^-_{j^N,2J_0+1-\kappa} \ .
				\end{equation}
Finally, the multiset $\mathpzc{E}^-_{j^N}$ is obtained as
				\begin{equation}
				\mathpzc{E}^-_{j^N} = \left\lbrace J_\kappa ^{\lambda^-_{j^N,\kappa}} \left| \lambda^-_{j^N,\kappa} > 0 \right. \right\rbrace \ .
				\end{equation}
For a given $N$ spin$-j$ system, we observe from Eq. \eqref{eq:symm_G_j_N_a} that the largest spin resulting from the symmetric composition is $J_0$, a result which could have been expected intuitively.
\par Considering now the spin-$\infty$ limit, it is immediate from Eq. \eqref{eq:symmetric_GENERAL_G_Omega} that
					\begin{equation}
					\label{eq:symmetric_G_infty_a}
					\lim_{j \to \infty} \mathscr{G}^-_{j,\Omega}(q) :=  \mathscr{G}^-_{\infty,\Omega}(q) = \frac{1}{\left(1-A \right)^{\infty}_q} = e^{A/(1-q)}_q
					\end{equation}
where $e^y_q$ (Euler's second identity) is another $q-$analogue of $e^x$ given by the expression\cite{book:Kac-2001},
					\begin{equation}
					\label{eq:Euler_second_id}
					e^y_q : = \sum^{\infty}_{k=0} \frac{y^k}{\left[k \right]_q!}  \ .
					\end{equation}
It is easy to verify that the generating function for $\left\lbrace \Omega^-_{\infty^N,n} \right\rbrace$, $G^-_{\infty^N,\Omega}(q)$, is of the form
					\begin{equation}
					\label{eq:symm_infty_limit_2}
					G^-_{\infty^N,\Omega}(q) = \frac{1}{(1-q)(1-q^2)\cdots(1-q^N)} = \sum^\infty_{n=0} \Omega^-_{\infty^N,n} \ q^n \ , 
					\end{equation}
from which we immediately have
					\begin{equation}
					\Omega^-_{\infty^N,n} = p_N(n) \ \quad \ \mbox{and} \ \quad \ \lambda^-_{\infty^N,n} = p_N(n) - p_N(n-1) \ .
					\end{equation}
Interestingly, the generating functions $\mathscr{G}^+_{\infty,\Omega}(q)$ and  $\mathscr{G}^-_{\infty, \Omega}(q)$ are related to each other through the equation,
					\begin{equation}
					 \mathscr{G}^-_{\infty,\Omega}(1/q) = \mathscr{G}^+_{\infty,\Omega}(q) 
					\end{equation}
where $\mathscr{G}^-_{\infty,\Omega}(1/q)$ must be understood as $e^{A/(1-q)}_{1/q}$.  In fact, the antisymmetric $\mathscr{G}^+_{\infty,\Omega}(q) $ is nothing but Euler's first identity, while the symmetric  $\mathscr{G}^-_{\infty,\Omega}(q)$ is identically Euler's second identity -- which are known to be related to each through the relation\cite{book:Kac-2001}: $e^{x}_{1/q} = E^x_q$. Additionally, from Eqs. \eqref{eq:antisymm_infty_limit_2} and \eqref{eq:symm_infty_limit_2} we readily have
					\begin{equation}
					G^+_{\infty^N,\Omega}(q)=q^{\binom{N}{2}}\ G^-_{\infty^N,\Omega}(q) \ .
					\end{equation}
Definitely, a similar relation exists between $G^+_{\infty^N,\lambda}(q)$ and $G^-_{\infty^N,\lambda}(q)$.
\par The symmetric and antisymmetric compositions have been discussed in \cite{art:Polychronakos-2016} on the basis of the grand partition function. Indeed, it is readily observed that if one sets $a \to e^{\mu \beta}$ ($\mu$ is the chemical potential of the system) and $x \to e^{-\beta'}$, $\mathscr{G}^{\pm}_{j,\Omega}(x)$ becomes the grand partition function of a spin-$j$ gas of noninteracting bosons and fermions, respectively\citep{art:Polychronakos-2016}. Due to a typo, the factor $q^{\binom{N}{2}}$ appearing in $G^{+}_{j,\Omega}(x)$ (and some subsequent conclusions) is missing though in the results of \citep{art:Polychronakos-2016}.

\subsection{Number theory: multi-restricted composition}\label{subsec:number_thoery}
The main results of this paper can also be applied to some very interesting problems in enumerative combinatorics, and this should come as no surprise. We briefly consider below the example of multi-restricted composition of a given integer.
\par Say $\mathscr{C}(n^{d_1}_1, \ldots, n^{d_\sigma}_\sigma; n)$ the number of compositions of the integer $n$ into at most $N = \sum^{\sigma}_{\alpha=1} d_\alpha$ parts, with $d_\alpha$ parts being at most $n_\alpha$ ($\alpha \in \{1, \ldots, \sigma\}$). For example, $\mathscr{C}(2^5, 4^3, 5^4; 16)$ is the number of compositions of the integer $16$ into at most twelve parts, with five of them being at most of value $2$, three being at most $4$ and four being at most $5$.
\par To determine $\mathscr{C}(n^{d_1}_1, \ldots, n^{d_\sigma}_\sigma; n)$, we need to make the following distinction:
\paragraph{\textbf{The number zero is an admissible part}.} In such event,  $\mathscr{C}(n^{d_1}_1, \ldots, n^{d_\sigma}_\sigma; n) \to \mathscr{C}_0(n^{d_1}_1, \ldots, n^{d_\sigma}_\sigma; n)$ (and reserve the former for the case whereby zero is not an admissible part).  In this case, $\mathscr{C}_0(n^{d_1}_1, \ldots, n^{d_\sigma}_\sigma; n)$ is exactly $\Omega_{\mathpzc{A},n}$ of the CGD of $N$ spins: $d_1$ being spin-$\frac{n_1}{2}$ representations of $SU(2)$, $d_2$ being spin-$\frac{n_2}{2}$, and so on. Thus, $\mathpzc{A}=\{\left(\frac{n_1}{2}\right)^{d_1}, \ldots , \left(\frac{n_\sigma}{2}\right)^{d_\sigma}\}$, and given that $\mathscr{C}_0(n^{d_1}_1, \ldots, n^{d_\sigma}_\sigma; n)=\Omega_{\mathpzc{A},n}$ we can use any of the three methods discussed above to determine 
$\mathscr{C}_0(n^{d_1}_1, \ldots, n^{d_\sigma}_\sigma; n)$. Resorting to the generating function approach for example, we can -- following Eq. \eqref{eq:gen_Omega_n_c} --   state that
				\begin{equation}
				\boxed{
				\prod^{\sigma}_{\alpha = 1} \left(\sum^{n_\alpha}_{i=0} q^i\right)^{d_\alpha} = \prod^{\sigma}_{\alpha = 1} \left(\left[n_\alpha+1\right]_q\right)^{d_\alpha}=  \sum^{\sum^{\sigma}_{\alpha=1} d_\alpha n_\alpha}_{n=0} \mathscr{C}_0(n^{d_1}_1, \ldots, n^{d_\sigma}_\sigma; n) \ q^n } \ .
				\end{equation}
For instance,
				\begin{equation}
				\left([3]_q\right)^5\left([5]_q\right)^3\left([6]_q\right)^4=\sum^{42}_{n=0} \mathscr{C}_0(2^5, 4^3, 5^4; n) \ q^n \ .
				\end{equation}
Certainly, it also follows from Eq. \eqref{eq:symmetric_Omega} that
				\begin{equation}
				\label{eq:multi_restricted_partition_id}
				\mathscr{C}_0(n^{d_1}_1, \ldots, n^{d_\sigma}_\sigma; n) = \mathscr{C}_0\left(n^{d_1}_1, \ldots, n^{d_\sigma}_\sigma; \sum^{\sigma}_{\alpha=1} d_\alpha n_\alpha - n\right) \ .
				\end{equation}
\par For the limit case $\mathscr{C}_0(c^d; n)$, it is not hard to see that $\mathscr{C}_0(c^d; n)=\Omega_{(c/2)^d,n}$, thus we may directly employ Eq. \eqref{eq:mono_spin_Omega_n} or \eqref{eq:mono_spin_Omega_n_hypergeo}. Indeed, based on the latter, we have 
				\begin{equation}
				\label{eq:p_0}
				\boxed{
				\mathscr{C}_0\left(c^d; n \right) =\binom{
				d+n-1}{
				n
				} \ \tensor[_{c+2}]{F}{_{c+1}} \left(\left. \myatop{-d, -\frac{n}{c+1},  \ldots ,- \frac{n-i}{c+1}, \ldots - \frac{n-c}{c+1}}{-\frac{d+n-1}{c+1}, \ldots ,- \frac{d+n-1-i'}{c+1}, \ldots - \frac{d+n-1-c}{c+1}}  \right\vert 1 \right)} \ .
				\end{equation}	
Needless to say, the generating function for $\{\mathscr{C}_0(c^d; n)\}$ is
				\begin{equation}
				\left(\left[ c+1\right]_q\right)^d = \sum^{dc}_{n=0} \mathscr{C}_0(c^d; n) \ q^n \ .
				\end{equation}
\paragraph{\textbf{Zero is not an admissible part}}	We can still map this case to the case in which zero is an admissible part. All we need to do is to make the following transformations: $n_\alpha \longrightarrow n_\alpha - 1$ , $n \longrightarrow n - N$. Therefore,
				\begin{equation}
				\label{eq:multi_restricted_partition_zero_not}
				\mathscr{C}(n^{d_1}_1, \ldots, n^{d_\sigma}_\sigma; n) = \mathscr{C}_0 \left( (n_1-1)^{d_1}, \ldots, (n_\sigma-1)^{d_\sigma}; n-N \right) \ .
				\end{equation}				
All we have done is just a rescaling of the integers. Moreover, for the limit case $\mathscr{C}(c^d;n)$, we can write
				\begin{equation}
				\label{eq:multi_restricted_partition_zero_not_b}
				\mathscr{C}(c^d;n) =\mathscr{C}_0 \left((c-1)^d;n-d\right)
				\end{equation}

Note that $n$ in Eqs. \eqref{eq:multi_restricted_partition_zero_not} and \eqref{eq:multi_restricted_partition_zero_not_b} is $N \leq n \leq \sum^{\sigma}_{\alpha=1} d_\alpha n_\alpha $.

\subsection{A simple statistical application: the rolling of $N$ dices}\label{subsec:dice}
The equations derived in the last section may find application in a number of statistical problems. We illustrate here just one.
Suppose we throw $N$ dices. We ask: what is the probability $P_6(N,n)$ that the sum of the outcomes is $n$? Certainly, 
				\begin{subequations}
				\begin{align}
				P_6(N,n) & = \frac{1}{6^{N}} \times \mathscr{C}(6^N;n) = \frac{1}{6^{N}} \times \mathscr{C}_0(5^N;n') \\
				& =\frac{1}{6^{N}} \times \binom{
				n-1 }{
				n'
				} \ \tensor[_{7}]{F}{_{6}} \left(\left. \myatop{-N, -\frac{n'}{6}, -\frac{n'-1}{6}  ,-\frac{n'-2}{6}, -\frac{n'-3}{6}, -\frac{n'-4}{6},- \frac{n'-5}{6}}{-\frac{n-1}{6}, -\frac{n-2}{6} ,-\frac{n-3}{6}, -\frac{n-4}{6},-\frac{n-5}{6}, -\frac{n-6}{6}}  \right\vert 1 \right) \ .
				\end{align}
				\end{subequations}
where $n' := n - N$, and obviously $N \leq n \leq 6N$, where use has been made of Eq. \eqref{eq:p_0}.
\par Consider now the coupling of $N$ spin-$5/2$s. $P_6(N,n)$ can be re-interpreted here as the probability that a randomly chosen $z-$eigenvalue in the coupled representation is of value $M_n$.

\section{Conclusion}
The seemingly unavailing effort of mapping the $z-$eigenvalues of $SU(2)$ spins to the set of natural numbers allows one to place the composition of an arbitrary multiset of $SU(2)$ spins into the context of enumerative combinatorics, as we have shown above. The striking gain here is a very general re-interpretation of the spin composition conundrum which allows one to solve the problem in diverse ways and in general terms. This is a feat hardly achievable without the $z-$eigenvalues-positive integers mapping. That this simplifies the analytic treatment of spin compositions is remarkable in its own right. But in retrospect, we note that the hallmark of quantum mechanics is the quantization of observables -- to which we can always associate a countable set or multiset of integers, in general. And this fact constitutes the very first instance of connection between quantum mechanics and discrete mathematics in general. Our work shows how and why EC can be a valid mathematical toolkit to be included in the arsenal of mathematical techniques employed in quantum mechanics. The proving of the minimum $J_m$ (sec. \ref{sec:min_max_J}), the three methods outlined in secs. \ref{sec:momentum_inv_sub_dim} and \ref{sec:multiplicities} for analytic CGD and the derivation of analytic symmetric and antisymmetric spin compositions (sec. \ref{subsec:symmteric_antisymmetric}) are all indicators of its relevance to quantum mechanics. The connection between EC and lattice paths may also offer interesting approaches to the study of spin dynamics, and bridge the gap between the latter and lattice gas models. We believe this to be an important effort worth pursuing.
\par Moreover, besides the applications discussed above, the analytical expressions presented here may find very broad applications in diverse fields -- from quantum algebra, quantum information, quantum statistical mechanics to group theory, not to mention quantum chemistry, and in particular quantum magnetic resonance (as we shall show in an upcoming article). The analytical method presented here may also be employed to study spin-orbit coupling in many-body systems.

\section*{Acknowledgment}
The authors thank Prof. J. Ricardo G. Mendonça (Universidade de São Paulo), Prof. Alexios Polychronakos (CUNY), Prof. Konstantinos Sfetsos (National and Kapodistrian Univ. of Athens), Prof. Maurizio Casarin (Univ. of Padova), Prof. Giorgio J. Moro (Univ. of Padova) and Prof. Antonino Polimeno (Univ. of Padova) for carefully reading an abridged version of this manuscript and for their insightful comments. Special thanks from JAG to Prof. J. Ricardo G. Mendonça and Mr. Andrea Piserchia (SNS). The authors employed computer facilities at SMART@SNS laboratory, to which they are greatly grateful. The research leading to these results has received funding from the European Research Council under the European Union's Seventh Framework Program (FP/2007-2013) / ERC Grant Agreement n. [320951].

\appendix
\section{Power series representation of $q-$binomial coefficients}
The $q-$binomial coefficient ${\ a \ \brack \ b \ }_q$ is defined as\cite{book:Kac-2001}:
				\begin{equation}
				\begin{split}
				{ \ a \ \brack \ b \ }_q & = \frac{[a]_q!}{[b]_q![a-b]_q!} = { \ a \ \brack \ a-b \ }_q  \\
				& = \frac{[a]_q[a-1]_q \cdots [a-b+1]_q}{[b]_q!} \\
				&= \frac{\left(q^a-1\right)\left(q^{a-1}-1\right)\cdots\left(q^{a-b+1}-1\right)}{\left(q^b-1\right)\left(q^{b-1}-1\right) \cdots \left(q-1\right)}
				\end{split}
				\end{equation}
for $0 \leq b \leq a$; while for $b > a$, ${\ a \ \brack \ b \ }_q=0$. It is of common knowledge that ${\ a \ \brack \ b \ }_q$ is a polynomial in $q$ of degree $b(a-b)$, and whose leading coefficient is $1$.
\par The power series representation of $q-$binomial coefficients is commonly obtained by resorting to combinatorial interpretations of the $q-$binomial coefficients. There may be many such combinatorial interpretations but the following two are widely known in the literature:
\begin{itemize}
\item \textsc{Interpretation 1} (see also \citep[Theorem 6.1, pg. 19]{book:Kac-2001}) Say $\mathbb{N}_{\leq a} :=\{1,2,\ldots , a\}$ and say $\mathbb{N}^{|b|}_{\leq a}$ the collection of all distinct subsets of $\mathbb{N}_{\leq a}$ of dimension $b$, $0 \leq b \leq a$. If we indicate the $i-$th element of $\mathbb{N}^{|b|}_{\leq a}$ as $\mathbb{N}^{|b|}_{\leq a,i}$, then $\mathbb{N}^{|b|}_{\leq a} = \left\lbrace \mathbb{N}^{|b|}_{\leq a,1}, \mathbb{N}^{|b|}_{\leq a,2}, \ldots , \mathbb{N}^{|b|}_{\leq a, \binom{a}{b}} \right\rbrace$. That is, the cardinality of $\mathbb{N}^{|b|}_{\leq a}$ is $\binom{a}{b}$. If we now take the sum of the elements of each $ \mathbb{N}^{|b|}_{\leq a,i} $, it can be easily proved that there are $\left[b(a-b)+1\right]$ distinct possible sums. We can represent the resulting sums by the multiset $S:= \left\lbrace S_0^{n_0}, S_1^{n_1}, \ldots , S^{n_{b(a-b)}}_{b(a-b)}  \right\rbrace$, where $n_0 + n_1 + \ldots + n_{b(a-b)}=\binom{a}{b}$ and 
				\begin{equation}
				S_l := \frac{b(b+1)}{2} + l \ , \ \qquad \ 0 \leq l \leq b(a-b) \ . 
				\end{equation}

 Then ${\ a \ \brack \ b \ }_q$ satisfies the relation
				\begin{subequations}
				\begin{align}
				{\ a \ \brack \ b \ }_q & = \sum_{S_l \in S} q^{S_l-b(b+1)/2} \label{eq:comb_int_1_a} \\
										& =\sum^{b(a-b)}_{l=0} n_l \  q^{l} \label{eq:comb_int_1_b} \ .
				\end{align}
				\end{subequations}
(For proof of Eq. \eqref{eq:comb_int_1_a}, see \citep[pg. 19]{book:Kac-2001}). Thus the power series representation of $q-$binomial reduces to determining the set of nonnegative integers $\{n_l\}$. There is no analytical formula in the literature, to the knowledge of the authors, which allows a direct calculation of these integers.
\item \textsc{Interpretation 2} Using Young diagrams, the integer $n_l$ in Eq. \eqref{eq:comb_int_1_b} may be interpreted as the number of distinct partitions of the nonnegative integer $l$ which fit into a $b \times (a-b)$ rectangle\cite{link:Weisstein_2}. Here again, to the best of the knowledge of the authors, there is no analytical expression to determine such a partition. In the following, we shall employ the same notation in \cite{book:Andrews-1976} for $n_l$, namely $p(a-b,b,l).$ 
\end{itemize}

$p(a-b,b,l)$ is the number of restricted partitions of $l$ into at most $b$ parts, each being at most $(a-b)$\citep{book:Andrews-1976}. In this section, we shall attempt to derive an analytical expression for the coefficients $\left\lbrace p(a-b,b,l) \right\rbrace$. 
We begin with the following theorem:
\begin{theorem}\label{thm:q_binomial_conv_sum_expansion}
For any pair of nonnegative integers $a,b$, with $a\geq b$, then
				\begin{equation}
				\label{eq:q_binomial_expansion_1}							
				\boxed{		
				{ \ a \ \brack \ b \ }_q  = 1 + \sum^{a-b}_{m_1=1}q^{m_1}  + \sum^{a-b}_{m_1=1}\sum^{m_1}_{m_2=1}q^{m_1+m_2}  + \cdots + \sum^{a-b}_{m_1=1}\sum^{m_1}_{m_2=1}\cdots \sum^{m_{b-1}}_{m_{b}=1} q^{m_1+m_2+m_3+\cdots + m_{b}} 
				 }
				\end{equation}
if $a>b$, while ${ \ a \ \brack \ b \ }_q =1$ if $a=b$.				
\end{theorem} 				

\begin{proof}
We begin with the $q-$Pascal rule\cite{book:Kac-2001}:
				\begin{equation}
				\label{eq:Pascal_rule_1}
				{ \ a \ \brack \ b \ }_q =  { \ a - 1 \ \brack \ b \ }_q + q^{a-b} { \ a-1 \ \brack \ b-1 \ }_q
				\end{equation}
and the fact that for any nonnegative integer $a$,
				\begin{equation}
				{ \ a \ \brack \ 0 \ }_q =  { \ a \ \brack \ a \ }_q = 1 \ ,
				\end{equation}
(from which follows the case $a=b$ of the theorem). 
\par For $a> b$, we repeat the above $q-$Pascal rule on the first term of Eq. \eqref{eq:Pascal_rule_1} $c$ times, where $1 \leq c \leq a-b$, and obtain
				\begin{equation}
				\label{eq:Pascal_rule_repeated_c}
				{ \ a \ \brack \ b \ }_q =  { \ a - c \  \brack \ b \ }_q + \sum^{c-1}_{k=0}q^{a-b-k} { \ a-1-k \ \brack \ b-1 \ }_q \ .
				\end{equation}
If we now set $c=a-b$, Eq. \eqref{eq:Pascal_rule_repeated_c} becomes
				\begin{equation}
				\label{eq:Pascal_rule_repeated_d}
				{ \ a \ \brack \ b \ }_q = 1 + \sum^{a-b-1}_{k=0}q^{a-b-k} { \ a-k-1 \ \brack \ a-b-k \ }_q \ .
				\end{equation}
which may also be written as,
				\begin{equation}
				\label{eq:Pascal_rule_repeated_e}
				{ \ a \ \brack \ b \ }_q = 1 + \sum^{a-b}_{m_1=1}q^{m_1} { \ b-1+m_1 \  \brack \ b-1 \ }_q \ .
				\end{equation}
Expanding ${ \ b-1+m_1 \brack \ b-1}_q$ according to Eq. \eqref{eq:Pascal_rule_repeated_e}  and inserting it back into the said equation, we obtain
				\begin{equation}
				{ \ a \ \brack \ b \ }_q = 1 + \sum^{a-b}_{m_1=1}q^{m_1}  + \sum^{a-b}_{m_1=1}\sum^{m_1}_{m_2=1}q^{m_1+m_2} { \ b-2+m_2 \ \brack \ b-2 \ }_q \ .
				\end{equation}
By induction, one can prove that the following equation holds after $n$ repetitions:
				\begin{equation}
				{ \ a \ \brack \ b \ }_q  = 1 + \sum^{a-b}_{m_1=1}q^{m_1}  + \sum^{a-b}_{m_1=1}\sum^{m_1}_{m_2=1}q^{m_1+m_2}  + \cdots + \sum^{a-b}_{m_1=1}\sum^{m_1}_{m_2=1}\cdots \sum^{m_{n-1}}_{m_{n}=1} q^{m_1+m_2+m_3+\cdots + m_{n}}{ \ b-n+m_{n} \  \brack \ b-n \ }_q \ .
				\end{equation}
The repetition certainly ends when ${ \ b-n+m_{n} \ \brack \ b-n \ }_q = 1$ for all $m_n$. This obviously happens when $n=b$, which proves the theorem.
\end{proof}

Furthermore, the following corollary follows from Theorem \ref{thm:q_binomial_conv_sum_expansion} and the property of $q-$binomial coefficients: 
\begin{corollary}
For the nonnegative integers $a,b$,
				\begin{equation}
				\boxed{
				{ \ a \ \brack \ b \ }_q  = 1 + \sum^{b}_{m_1=1}q^{m_1}  + \sum^{b}_{m_1=1}\sum^{m_1}_{m_2=1}q^{m_1+m_2}  + \cdots + \sum^{b}_{m_1=1}\sum^{m_1}_{m_2=1}\cdots \sum^{m_{a-b-1}}_{m_{a-b}=1} q^{m_1+m_2+m_3+\cdots + m_{a-b}} }
				\end{equation}
where $a \geq b$.
\end{corollary}

We can finally state the following theorem -- closely related to Theorem \ref{thm:q_binomial_conv_sum_expansion} -- which gives an analytical expression for the coefficients of the power series of Gauss polynomials:
\begin{theorem}\label{thm:q_binomial_ps_explicit}
Given the nonnegative integers $a \geq b$, then 
				\begin{equation}
				{ \ a \ \brack \ b \ }_q = \sum^{b(a-b)}_{k=0} \ p(a-b,b,k) \ q^k
				\end{equation}
where,
				\begin{equation}
				\label{eq:def_c_a_b_k}
				\boxed{
				 p(a-b,b,k) = \sum^b_{\nu=0} \phi^{a,b}_{\nu,k} } \ ,
				\end{equation}
with
				\begin{subequations}
				\label{eq:phi_a_bs}
				\begin{align}
				 \phi^{a,b}_{0,k} & = \delta_{0,k} \label{eq:phi_a_bs_0} \\
				 \phi^{a,b}_{1,k} & = H(a-b-k) \ H(k-1) \label{eq:phi_a_bs_1} \\
				 \phi^{a,b}_{2,k} & = \sum^{a-b}_{m_1=1}H(2m_1-k)\ H(k-m_1-1) \label{eq:phi_a_bs_2} \\
				 \phi^{a,b}_{3,k} &  =\sum^{a-b}_{m_1=1}\sum^{m_1}_{m_2=1}H(2m_2+m_1-k)\ H(k-m_1-m_2-1) \label{eq:phi_a_bs_3}\\
				 				  & \vdots \nonumber \\
				 \phi^{a,b}_{\nu,k} &  =\sum^{a-b}_{m_1=1} \cdots \sum^{m_{\nu-2}}_{m_{\nu-1}=1}H(m_1+\ldots + m_{\nu-2}+2m_{\nu-1}-k)\ H(k-1-m_1-\ldots -m_{\nu-1}) \label{eq:phi_a_bs_k}
				\end{align}
				\end{subequations}
and where $H(x)$ is the Heaviside step-function defined in Eq. \eqref{eq:Heaviside}.
\end{theorem}

\begin{proof}
Eq. \eqref{eq:q_binomial_expansion_1} may be rewritten as
				\begin{equation}
				\label{eq:q_binomial_expansion_2}
				\begin{split}
				{ \ a \ \brack \ b \ }_q & = 1 + \sum^{a-b}_{\mu=1} \phi^{a,b}_{1,\mu} \ q^{\mu}  + \sum^{2(a-b)}_{\mu=2} \phi^{a,b}_{2,\mu} \ q^{\mu}  + \cdots + \sum^{b(a-b)}_{\mu=b} \phi^{a,b}_{b,\mu} \ q^{\mu} \\
				 & =  \sum^{b(a-b)}_{k=1} \sum^b_{\nu=0}\phi^{a,b}_{\nu,k} \ q^{k} \ ,
				\end{split}
				\end{equation}
The $\{\phi^{a,b}_{\nu,\mu}\}$ are peculiar partition functions of the nonnegative integer $k$: In particular, $\phi^{a,b}_{\nu,\mu}$ counts the number of distinct ways of obtaining $k$ from the sum of exactly $\nu$ integers: $m_1 + m_2 + \ldots + m_\nu = k$, where each $m_i \ (\in \mathbb{N})$ is at most $\nu(a-b)$ and $m_1 \geq m_2 \geq \ldots \geq m_\nu$.  Eq. \eqref{eq:def_c_a_b_k} readily follows from the last equation on the RHS of Eq. \eqref{eq:q_binomial_expansion_2}. Note that the set of integers with nonzero partitions according to $\phi^{a,b}_{\nu,k}$ is bounded from below by $\nu$ and above by $\nu(a-b)$.
\par Regarding the explicit form of the various $\phi^{a,b}_{\nu,k}$, it is clear that Eq. \eqref{eq:phi_a_bs_0} holds. When it comes to $\phi^{a,b}_{1,k}$, we are talking about a one part partition of $k$, with the part not greater than $(a-b)$, \emph{i.e.} 
				\begin{equation}
				\label{eq:phi_1_partion_equality}
				\phi^{a,b}_{1,k}=p(a-b,1 ,k)
				\end{equation}
where $p(n,1,k)$ is the partition function of $k$ into a single part not greater than $n$ (and in which zero is not an admissible part). It is obvious that 
				\begin{equation}
				\label{eq:partion_one_Heaviside_identity}
				p(n,1,k)=H(n-k) \ H(k-1) \ , \qquad \ \forall \ n,k \in \mathbb{N}
				\end{equation}
where $H(x)$ is the Heaviside step-function defined as in Eq. \eqref{eq:Heaviside}. 
From Eqs. \eqref{eq:phi_1_partion_equality} and \eqref{eq:partion_one_Heaviside_identity} follows \eqref{eq:phi_a_bs_1}. 
\par For $\phi^{a,b}_{2,k}$ we are considering the distinct partitions of $k$ into two parts, \emph{i.e.} $m_1+m_2 =k$, with $m_1 \geq m_2$. We note from Eq. \eqref{eq:q_binomial_expansion_1} that whenever we fix $m_1$, the maximum value $m_2$ can take is $m_1$. This implies that upon fixing $m_1$, the partition of $k$ is simply $p(m_1,1, k-m_1)$, which upon summing over all the allowed values of $m_1$ and employing the identity in Eq. \eqref{eq:partion_one_Heaviside_identity} yields Eq. \eqref{eq:phi_a_bs_2}. For any generic $\phi^{a,b}_{\nu,k}$, one imagines fixing the $m_1,\ldots, m_{\nu-1}$ which leaves out only $m_\nu$ as the only degree of freedom; but since the sum $m_1 + \ldots + m_\nu=k$ must hold and the maximum value of $m_\nu$ is $m_{\nu-1}$, the partition $\phi^{a,b}_{\nu,k}$ reduces to $p(m_{\nu-1},1,k-m_1-\ldots - m_{\nu-1})$, which upon applying Eq. \eqref{eq:partion_one_Heaviside_identity} and summing over $m_1 , \ldots, m_{\nu_1}$ gives Eq. \eqref{eq:phi_a_bs_k}.
\end{proof}
\par The partition functions $\left\lbrace \phi^{a,b}_{\nu,k} \right\rbrace$ in Eq. \eqref{eq:phi_a_bs} can be written in a more explicit form by performing the summation. For $\phi^{a,b}_{1,k}$ it is immediate from Eq. \eqref{eq:phi_a_bs_1} that
					\begin{equation}
					\phi^{a,b}_{1,k} = 
					\begin{cases}
					1 \ , & \mbox{if } 1\leq k \leq a-b  \\
					0 \ , & \mbox{otherwise}
					\end{cases}
					\end{equation}
For $\nu \geq 2$, $\phi^{a,b}_{\nu,k}$ is not that immediate, though the process involves very elementary mathematics. To illustrate this point, we write $\phi^{a,b}_{2,k}$ in a more explicit form without the summation over $m_1$. In order to do so, we make use of the identity
					\begin{equation}
					\label{eq:Heaviside_Sgn_relation}
					H(x) = \frac{1}{2} \left[ 1 + \Sgn(x) \right]
					\end{equation}					 
where $\Sgn(x)$ is the modified sign function, defined as
					\begin{equation}
					\label{eq:def_Sgn_function}
					\Sgn(x) := \begin{cases}
					-1 \ , & \mbox{if } x < 0 \\
					+1 \ , & \mbox{if } x \geq 0
					\end{cases} \ .
					\end{equation}
Now, in light of Eq. \eqref{eq:Heaviside_Sgn_relation}, we may rewrite Eq. \eqref{eq:phi_a_bs_2} as
					\begin{equation}
					\label{eq:phi_a_b_2_expanded_i}
					\phi^{a,b}_{2,k}= \frac{1}{4}\sum^{a-b}_{m_1=1} \left[1+\Sgn(k-m_1-1) + \Sgn(2m_1-k) +  \Sgn(k-m_1-1)\cdot \Sgn(2m_1-k) \right] \ .
					\end{equation}
To proceed, we must consider separately the events: 1) $a-b >k $, 2) $a-b < k$ and 3) $a-b=k$. Moreover, for each case, it is convenient to rewrite Eq. \eqref{eq:phi_a_b_2_expanded_i} in such a way that each summand appearing on the RHS has a unique sign.
\par If $a-b > k$, then Eq. \eqref{eq:phi_a_b_2_expanded_i} may be rewritten as
					\begin{multline}
					\phi^{a,b}_{2,k}= \frac{1}{4} \left[(a-b)+ \sum^{k-1}_{m_1=1}\Sgn(k-m_1-1) + \sum^{a-b}_{m_1=k}\Sgn(k-m_1-1) + \sum^{\lfloor \frac{k-1}{2} \rfloor}_{m_1=1} \Sgn(2m_1-k) \right. \\ 
					 + \sum^{a-b}_{m_1=\lfloor \frac{k-1}{2} \rfloor +1} \Sgn(2m_1-k) + \sum^{\lfloor \frac{k-1}{2} \rfloor}_{m_1=1} \Sgn(2m_1-k) \cdot \Sgn(k-m_1-1) \\
					\left. + \sum^{k-1}_{m_1=\lfloor \frac{k-1}{2} \rfloor +1}\Sgn(2m_1-k) \cdot \Sgn(k-m_1-1) + \sum^{a-b}_{m_1=k}\Sgn(2m_1-k) \cdot \Sgn(k-m_1-1)\right] 	\ ,		
					\end{multline}
which upon summation yields
					\begin{equation}
					\label{eq:phi_2_k_a}
					\phi^{a,b}_{2,k} = k-1 - \left\lfloor \frac{k-1}{2} \right\rfloor \ , \qquad \mbox{if } (a-b) >k \ .
					\end{equation}
The result proves that when $(a-b) >k$, the partition $\phi^{a,b}_{2,k}$ is independent of $a$ and $b$.
\par In the case whereby $a-b < k$, $\Sgn(k-m_1-1)$ is always positive. But for $\Sgn(2m_1-1)$ we must consider separately the cases: i) $a-b > \left\lfloor \frac{k-1}{2}\right\rfloor$, ii) $a-b < \left\lfloor \frac{k-1}{2}\right\rfloor$ and iii) $a-b = \left\lfloor \frac{k-1}{2}\right\rfloor$. Following the same lines of reasoning as in the case $a-b >k$ discussed above, one gets -- after some algebra -- 
					\begin{equation}
					\label{eq:phi_2_k_b}
					\phi^{a,b}_{2,k} = 
					\begin{cases}
					0 \ , & \mbox{if } a-b \leq \left\lfloor \frac{k-1}{2} \right\rfloor \\
					a-b - \left\lfloor \frac{k-1}{2} \right\rfloor \ , & \mbox{if }\left\lfloor \frac{k-1}{2} \right\rfloor < a-b < k 
					\end{cases} \ .
					\end{equation}	
\par For $a-b=k$, we derive that
					\begin{equation}
					\label{eq:phi_2_k_c}
					\phi^{a,b}_{2,k} = a-b-1 - \left\lfloor \frac{a-b-1}{2} \right\rfloor \ .
					\end{equation}	
Putting together Eqs. \eqref{eq:phi_2_k_a}-\eqref{eq:phi_2_k_c}, we have that				
					\begin{equation}
					\label{eq:phi_2_k_final}
					\phi^{a,b}_{2,k} = 
					\begin{cases}
					k-1 - \left\lfloor \frac{k-1}{2} \right\rfloor \ , & \mbox{if } k\leq (a-b)  \\
					a-b - \left\lfloor \frac{k-1}{2} \right\rfloor \ , & \mbox{if }\left\lfloor \frac{k-1}{2} \right\rfloor < (a-b) < k \\
0 \ , & \mbox{if } \left\lfloor \frac{k-1}{2} \right\rfloor \geq (a-b) 
					\end{cases} \ .					
					\end{equation}	

\par We remark that,
					\begin{equation}
					\phi^{a,b}_{\nu,k} = \phi^{a',b'}_{\nu',k} \ \qquad \ (a\geq b \ \wedge a' \geq b')
					\end{equation}
if $a-b=a'-b'$, $\nu = \nu'$ and $\nu \leq \min(b,b')\cdot (a-b)$. That is, the partitions $\phi^{a,b}_{\nu,k}$ are translationally invariant with respect to fixed $\nu$ as long as $\nu \leq \min(b,b')\cdot (a-b)$. 			
\par A number of identities and recurrence relations involving the restricted partitions $p(n,m,k)$ (or even $\phi^{a,b}_{\nu,k}$) can be derived from the above relations. Among these are
					\begin{subequations}
					\label{eq:recurrence_p}
					\begin{align}
					p(a-b,b,k) & = p(b,a-b,k) \label{eq:recurrence_p_1} \\
					p(a-b,b,k) & = p(a-b,b-1,k) + p(a-b-1,b,k-b) \label{eq:recurrence_p_2} \\
					p(a-b,b,k) & = p(a-b,b-1,k+b-a) + p(a-b-1,b,k) \label{eq:recurrence_p_3} \\
					p(a-b,b,k) & = \sum^{a-b}_{\substack{m=1 \\ (m \leq k \leq bm)}} p(m,b-1,k-m) \label{eq:recurrence_p_4}
					\end{align}
					\end{subequations}	
The corresponding relations for $\phi^{a,b}_{\nu,k}$ may be derived from Eq. \eqref{eq:def_c_a_b_k}. For instance, combining the latter and Eq. \eqref{eq:recurrence_p_1} yields the identity
					\begin{equation}
					\sum^k_{\nu=0} \phi^{a,b}_{\nu,k}  = \sum^k_{\nu=0} \phi^{b,a-b}_{\nu,k} \ .
					\end{equation}	
Eqs. \eqref{eq:recurrence_p_1} and \eqref{eq:recurrence_p_2} are well-known identities (see for example \citep[Chap. 3]{book:Andrews-1976}), but it is worth noting that they have been derived here, without much effort, using essentially properties of Gaussian polynomials.  
\par Finally, we remark that in the limit $q \to 1$, Eq. \eqref{eq:q_binomial_expansion_1} turns out to be
					\begin{equation}
					\binom{a}{b}= \sum^b_{k=0} \binom{a-b-1+n}{n} \ ,
					\end{equation}
while from Eq. \eqref{eq:q_binomial_expansion_2} we have
					\begin{equation}
					\binom{a}{b}= \sum^{a-b}_{k=0} \binom{b-1+n}{n} \ ,
					\end{equation}
which are well-known binomial identities.

\bibliography{Gen_Clebsch_Gordan_biblio_2}
\end{document}